\documentclass[letter,11pt]{article}

\usepackage{amsmath,amsthm,amssymb}
\usepackage[linesnumbered,lined,boxed,commentsnumbered,ruled,vlined]{algorithm2e}
\usepackage{booktabs}
\usepackage{array}
\usepackage{graphicx}
\usepackage[margin=1in]{geometry}
\usepackage[all]{xy}
\usepackage{subfiles}
\usepackage{enumerate}
\usepackage[shortlabels]{enumitem}
 \setlist{noitemsep, topsep=3pt, parsep=2pt}
\usepackage{csquotes}
\usepackage{titling}
\usepackage{multirow}
\usepackage{titlesec}
\titlespacing*{\section}{0ex}{2ex minus .5ex}{1.5ex}
\titlespacing*{\subsection}{0ex}{1.5ex minus .3ex}{1.2ex}
\usepackage[bookmarks=true,hypertexnames=false]{hyperref}
\hypersetup{colorlinks=true, citecolor=blue, linkcolor=red, urlcolor=blue}
\usepackage[perpage]{footmisc}
\makeatletter

\newcommand{\Rmnum}[1]{\expandafter\@slowromancap\romannumeral #1@}
\makeatother

\newcommand{\todo}[1]{\typeout{TODO: \the\inputlineno: #1}\textbf{[[[ #1 ]]]}}

\newtheoremstyle{exampstyle}
{4pt} 
{3pt} 
{\itshape} 
{} 
{\bfseries} 
{.} 
{.5em} 
{} 

\theoremstyle{exampstyle} 
\newtheorem{theorem}{Theorem}[section]

\newtheorem{lemma}[theorem]{Lemma}

\newtheorem{definition}[theorem]{Definition}

\newenvironment{remark}{{\bfseries \noindent {Remark}.}}{\par\smallskip}{\par\smallskip}

\begin{document}

\title{{\bf Contraction: a Unified Perspective of Correlation Decay and Zero-Freeness of 2-Spin Systems}}
\author{
Shuai Shao\thanks{Department of Computer Sciences, University of Wisconsin-Madison. Supported by NSF CCF-1714275}\\ 
{\tt  sh@cs.wisc.edu}
\and {Yuxin Sun}\thanksmark{1}\\
\tt {yxsun@cs.wisc.edu}}

\date{}
\clearpage
\maketitle
\thispagestyle{empty}
\begin{abstract}
We study complex zeros of the partition function of 2-spin systems, 
viewed as a multivariate polynomial in terms of the edge interaction parameters and the uniform external field.
We obtain new zero-free regions in which  all these parameters are complex-valued.
Crucially based on the zero-freeness, we show the existence of 
correlation decay in these complex regions.
As a consequence, 
we obtain an FPTAS 
for 
computing 
the partition function of 
2-spin systems on graphs of bounded degree 
for these parameter settings.
We introduce the \emph{contraction} property 
as a \emph{unified} sufficient condition to devise FPTAS via either Weitz's algorithm or Barvinok's algorithm. 
Our main technical contribution is 
a very simple but general approach to
extend any \emph{real} parameter  of which the 
2-spin system exhibits correlation decay
to its \emph{complex}  neighborhood where the partition function is zero-free and correlation decay still exists.
This result formally establishes the inherent connection between two distinct notions of phase transition for 2-spin systems: 
the existence of 
correlation decay and the zero-freeness 
of the partition function
via a {unified} perspective, {contraction}.
\end{abstract}

\newpage
\clearpage
\setcounter{page}{1}

\section{Introduction}
Spin systems originated from statistical physics to model interactions between
neighbors on graphs. 
In this paper, we focus on 2-state spin (2-spin) systems.
Such a system  is specified by two edge interaction parameters $\beta$ and $\gamma$, and a uniform external field $\lambda$.
An instance is a graph $G=(V, E)$.
A configuration $\sigma$ is a mapping $\sigma: V \rightarrow \{+, -\}$ which assigns one of the two spins $+$ and $-$ to each vertex in $V$. 
The weight $w(\sigma)$ of a configuration $\sigma$ is given by 
$$w(\sigma)=\beta^{m_+(\sigma)}\gamma^{m_{-}(\sigma)}\lambda^{n_+(\sigma)},$$
where $m_+(\sigma)$ denotes the number of $(+, +)$ edges under the configuration $\sigma$, 
$m_-(\sigma)$ denotes the number of $(-, -)$ edges, 
and $n_+({\sigma})$ denotes the number of vertices assigned to spin $+$. 
The partition function $Z_G(\beta, \gamma, \lambda)$ of the system parameterized by $(\beta, \gamma, \lambda)$
is defined to be the sum of weights over all configurations, i.e.,
$$Z_G(\beta, \gamma, \lambda)=\sum_{\sigma: V \rightarrow \{+, -\}}w(\sigma).$$
It is a sum-of-product computation.
If a 2-spin system is restricted to graphs of degree bounded by $\Delta$,
we say such a system is $\Delta$-bounded.

In classical statistical mechanics the parameters $(\beta, \gamma, \lambda)$ are usually non-negative real numbers,
and such 2-spin systems are divided into \emph{ferromagnetic} case ($\beta\gamma>1$) and \emph{antiferromagnetic} case ($\beta\gamma <1$).
The case $\beta\gamma =1$ is degenerate. 
When $(\beta, \gamma, \lambda)$ are  non-negative numbers and they are not all zero,
the partition function can be viewed as the normalizing factor of the Gibbs distribution, 
which is the distribution where a configuration $\sigma$ is drawn with probability 
${\rm Pr}_{G; \beta, \gamma, \lambda}(\sigma)=\frac{w(\sigma)}{Z_G(\beta, \gamma, \lambda)}$.
 However, it is meaningful to consider parameters of complex values. 
 By analyzing the location of \emph{complex} zeros of the partition function, 
 the phenomenon of \emph{phase transitions} was defined by physicists. 
 One of the first and also the best known result is the Lee-Yang theorem \cite{li-yang} for the \emph{Ising model}, a special case of 2-spin systems.
 This result was later extended to more general models by several people \cite{asano70, rue71, sg73, new74, ls81}.
 In this paper, we view the partition function $Z_G(\beta, \gamma, \lambda)$
 as a \emph{multivariate} polynomial over these three \emph{complex} parameters $(\beta, \gamma, \lambda)$. We study the zeros of this polynomial and the relation to the approximation of the partition function.

Partition functions encode rich information about the macroscopic properties of 2-spin systems. 
They are not only  of significance in statistical physics, 
but also are well-studied in computer science.
Computing the partition function of 2-spin systems given an input graph $G$ can be viewed as the most basic case of  Counting Graph Homomorphisms (\#GH)  \cite{dyer-green-gh, bg-gh-nonnegative,ggjt-gh-real, cai-chen-lu-gh} and
Counting Constraint Satisfaction Problems (\#CSP) \cite{dyer-ann-jerrum-weighted-csp, clx-boolean-csp-complex, bulatov-ccsp, dyer-richerby, cai-chen-csp-complex}, which are two very well studied frameworks for counting problems.
Many natural combinatorial problems can be formulated as 2-spin systems. 
For example, when $\beta=\gamma$, such a system is the famous \emph{Ising model}.
And when $\beta = 0$ and $\gamma=1$,  $Z_G(0, 1, \lambda)$ is the independence polynomial of the graph $G$
(also known as the \emph{hard-core model} in statistical physics); it counts the number of independent 
sets of the graph $G$ when $\lambda =1$.
\subsection*{Related work}
For exact computation of $Z_G(\beta, \gamma, \lambda)$, 
the problem is proved to be \#P-hard for all complex valued parameters
but a few very restricted trivial settings \cite{bar82, cai-chen-lu-gh, clx-boolean-csp-complex}.
So the main focus is to approximate $Z_G(\beta, \gamma, \lambda)$. 
This is an area of active research, and many inspiring algorithms are developed.
The pioneering algorithm developed by Jerrum and Sinclair 
 gives a \emph{fully polynomial-time randomized approximation scheme} (FPRAS) for the {ferromagnetic}  Ising model \cite{js}. 
 This FPRAS is based on the
\emph{Markov Chain Monte Carlo} (MCMC) method  which devises approximation counting algorithms via random sampling.
Later, it was extended to general ferromagnetic 2-spin systems \cite{gjp03, llz14}. 
 The MCMC method can only handle non-negative parameters as it is based on probabilistic sampling.



 The \emph{correlation decay} method developed by Weitz \cite{weitz}  was originally used to devise \emph{deterministic fully polynomial-time approximation schemes} (FPTAS) for the hardcore model up to the uniqueness threshold.
 It turns out to be a very powerful tool for devising FPTAS for antiferromagnetic 2-spin systems \cite{zhangbai, lly12, lly13,  sst14}.
Combining with hardness results \cite{ss14, gsv16}, an exact threshold of computational complexity transition of antiferromagnetic 2-spin systems 
is identified and the only remaining case is at the critical point.
On the other hand, for ferromagnetic 2-spin systems, limited results \cite{zhangbai, gl18} have been obtained via the correlation decay method.
Although correlation decay is usually analyzed in 2-spin systems of non-negative parameters,
it can be adapted to complex parameters. 
An FPTAS was obtained for the hard-core model in the  Shearer's region (a disc in the complex plane) via correlation decay in \cite{hsv18}.

Recently, a new method developed by Barvinok \cite{bar16}, and extended by Patel and Regts \cite{pr17} is the \emph{Taylor polynomial interpolation} method
that turns complex zero-free regions of the partition function into 
FPTAS of corresponding complex parameters. 
Suppose that the partition function $Z_G(\beta, \gamma, \lambda) $ has no zero in a complex region containing an easy computing point, e.g., $\lambda=0$. 
It turns out that, probably after a change of coordinates, 
 $\log Z_G(\beta, \gamma, \lambda)$ is well approximated
in a slightly smaller region by a low degree Taylor polynomials  which can be efficiently computed. 
This method connects the long-standing study of complex zeros  to algorithmic studies of the partition function of physical systems.
Motivated by this, more recently some  complex zero-free regions have been obtained for hard-core models \cite{bc18,pr19}, Ising models \cite{lss19a}, and general 2-spin systems \cite{gll19}. 


\subsection*{Our contribution}
In this paper, we obtain new zero-free regions of the partition function of 2-spin systems. 
Crucially based on the zero-freeness, we show the existence of 
correlation decay in these \emph{complex} regions.
As a consequence, 
we obtain an FPTAS 
for 
computing 
the partition function of 
bounded 2-spin systems
for these parameter settings.
Our result gives the first zero-free regions in which all three parameters $(\beta, \gamma, \lambda)$ are complex-valued and  new correlation decay results for bounded ferromagnetic 2-spin systems.
Our main technical contribution is a very simple but general approach to extend any real parameter  of which the bounded 2-spin system exhibits correlation decay
to its complex  neighborhood where the partition function is zero-free and correlation decay still exists.  
We show that for bounded 2-spin systems,  the   \emph{real contraction}\footnote{See Dedinition \ref{def:real-contract}. In many cases, the existence of correlation decay boils down to this property.} property that ensures correlation decay exists for certain real parameters directly implies the zero-freeness and the existence of correlation decay of corresponding complex neighborhoods.

We formally describe our main result. We use $\pmb\zeta \in \mathbb{C}^3$ to denote the parameter vector $(\beta, \gamma, \lambda)$.
Since the case $\beta=\gamma=0$ is trivial,  by symmetry we always assume $\gamma \neq 0$ in this paper.
\begin{theorem}
Fix $\Delta \in \mathbb{N}$. If $\pmb\zeta_0\in \mathbb{R}^3$ satisfies real contraction for $\Delta$, then there exists a $\delta>0$ such that for any $\pmb \zeta \in \mathbb{C}^3$ where $\|\pmb\zeta-\pmb\zeta_0\|_{\infty}<\delta$, we have 
\begin{itemize}
    \item $Z_{G}(\pmb\zeta)\neq 0$ for every graph $G$\footnote{This is true even if $G$ contains arbitrary number of vertices pinned by a \emph{feasible configuration} (Definition \ref{def:feasible-configuration}).} of degree at most $\Delta$;
    \item the $\Delta$-bounded 2-spin system specified by $\pmb\zeta$ 
    exhibits correlation decay.
\end{itemize}
As a consequence, there is an FPTAS for computing $Z_{G}(\pmb\zeta)$.
\end{theorem}
This result formally establishes the inherent connection between two distinct notions of phase transition for bounded 2-spin systems: the existence of correlation decay and the zero-freeness of the partition function, via a unified perspective, \emph{contraction}.
The connection from  the existence of correlation decay of real parameters to the zero-freeness of corresponding complex neighborhoods 
was already observed for the hard-core model \cite{pr19} and the Ising model without external field \cite{lss19a}.
In this paper, we extend it to general 2-spin systems, and furthermore we establish the connection from the zero-freeness of complex neighborhoods back to the existence of correlation decay of such  complex regions.


Now, we give our zero-free regions. We first identify the sets of real parameters of which bounded 2-spin systems exhibit correlation decay.


\begin{definition}\label{Correaltion-decay-sets}
Fix integer $\Delta \geq 3$. We have the following four  sets where correlation decay exists.
\begin{enumerate}
    \item $\mathcal{S}^{\Delta}_1=\{\pmb \zeta\in \mathbb{R}^3\mid  \frac{\Delta-2}{\Delta}<\sqrt{\beta\gamma}<\frac{\Delta}{\Delta-2}, \beta, \gamma>0 \text{ and } \lambda \geq 0\}$,
    \item $\mathcal{S}^{\Delta}_2=\{ \pmb \zeta\in \mathbb{R}^3\mid \beta\gamma<1, \beta\geq 0, \gamma >0, \lambda \geq 0, \text{ and }  \pmb \zeta$ is up-to-$\Delta$ unique (see Definition \ref{def:uniqueness})\},
    \item $\mathcal{S}^{\Delta}_3=\{\pmb \zeta\in \mathbb{R}^3\mid  {\beta\gamma}>\frac{\Delta}{\Delta-2}, \beta, \gamma>0 \text{ and } 0\leq \lambda < \frac{\gamma}{t^{\Delta-1}[(\Delta-2)\beta\gamma-\Delta]}\}$ where $t=\max\{1, \beta\}$, and
    \item $\mathcal{S}^{\Delta}_4=\{\pmb \zeta\in \mathbb{R}^3\mid  {\beta\gamma}>\frac{\Delta}{\Delta-2}, \beta, \gamma>0 \text{ and } \lambda > \frac{(\Delta-2)\beta\gamma-\Delta}{\beta r^{\Delta-1}}\}$ where $r=\min\{1,  1/\gamma\}$.
\end{enumerate}
When context is clear, we omit the superscript $\Delta$.
\end{definition}
The set $\mathcal{S}^{\Delta}_1$ was given in \cite{zhangbai} and $\mathcal{S}^{\Delta}_2$ was given in \cite{lly13}.
 To our best knowledge,  $\mathcal{S}^{\Delta}_1$ and $\mathcal{S}^{\Delta}_2$  cover all \emph{non-negative} parameters of which \emph{bounded} 2-spin systems are known to exhibit correlation decay. 
 The sets $\mathcal{S}^{\Delta}_3$ and $\mathcal{S}^{\Delta}_4$ are obtained in this paper. 
 They give {new} correlation decay results and hence FPTAS for bounded ferromagnetic 2-spin systems\footnote{When $\beta<\gamma$ and $\lambda$ is sufficiently large, it is known that approximating the partition function of ferromagnetic 2-spin systems over general graphs is \#BIS-hard \cite{llz14}. 
Our result $\mathcal{S}^{\Delta}_4$ shows that there is an FPTAS for such a problem when restricted to graphs of bounded degree. When $\beta<1<\gamma$, the FPTAS obtained from $\mathcal{S}^{\Delta}_3$ is covered by \cite{gl18}.}.  

\begin{theorem}\label{the-intro-main}Fix integer $\Delta\geq 3$.
For every $\pmb \zeta_0\in \mathcal{S}^{\Delta}_i$ $(i \in [4])$, there exists a $\delta>0$ such that for any $\pmb \zeta \in \mathbb{C}^3$ where $\|\pmb\zeta-\pmb\zeta_0\|_{\infty}<\delta$, we have 
\begin{itemize}
    \item $Z_{G}(\pmb\zeta)\neq 0$ for every graph $G$  
     of degree at most $\Delta$; ($G$ may contain a feasible configuration.)
    \item the $\Delta$-bounded 2-spin system specified by $\pmb\zeta$ exhibits correlation decay.
\end{itemize}
Then via either Weitz's algorithm or Barvinok's algorithm, 
there is an FPTAS for computing $Z_{G}(\pmb\zeta)$.
\end{theorem}
\begin{remark}
The choice of $\delta$ does not depend on the size of the graph, only on $\Delta$ and $\pmb \zeta_0$.
\end{remark}

\subsection*{Organization}
This paper is organized as follows. 
In Section \ref{sec2}, we briefly describe Weitz's algorithm \cite{weitz}. We introduce real contraction as a sufficient condition for the existence of correlation decay of real parameters, and we show sets $\mathcal{S}^{\Delta}_i (i\in [4])$ satisfy it.
In Section \ref{sec3}, we briefly describe Barvinok's algorithm \cite{bar16}. 
We introduce complex contraction as a generalization of real contraction, and 
we show that it gives a unified sufficient condition for both the zero-freeness of the partition function and the existence of correlation decay of complex parameters. 
Finally, in Section \ref{sec4}, we prove our main result that real contraction implies complex contraction.
This finishes the proof of Theorem \ref{the-intro-main}.
We use the following diagram (Figure \ref{fig:structure}) to summarize our approach to establish the connection between correlation decay and zero-freeness.
We expect it to be further explored for other models. 

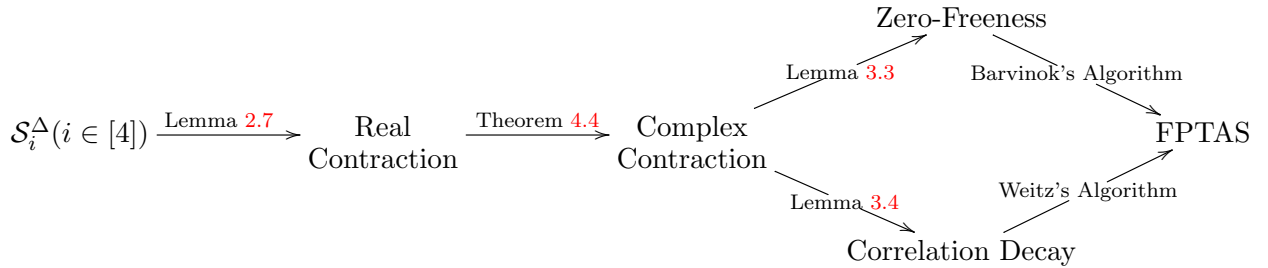
\begin{figure}[!htbp]
\vspace{-4ex}
$$\xymatrix{
 && && & \text{Zero-Freeness} \ar[dr]|-{\text{Barvinok's Algorithm}}  & \\
 \mathcal{S}^{\Delta}_i (i \in [4]) \ar[rr]^{\hspace{-2ex}\text{Lemma \ref{lem:c1c2}}}&& \txt{Real \\Contraction}\ar[rr]^{\text{Theorem \ref{real-to-complex}}} && \txt{{Complex}\\Contraction } \ar[ur]|-{\text{Lemma \ref{complex-to-zerofree}}}\ar[dr]|-{\text{Lemma \ref{lem:complex-to-correlation}}}  &  &  \text{FPTAS} \\
&& && & \text{Correlation Decay} \ar[ur]|-{\text{Weitz's Algorithm}}&
}$$
\vspace{-4.5ex}
\caption{The structure of our approach}
\label{fig:structure}
 \end{figure}

\subsection*{Independent work}
After a preliminary version \cite{v1} of this manuscript was posted, we learned that based on similar ideas, Liu simplified the proofs of \cite{pr19} and \cite{lss19a} and generalized them to antiferromagnetic Ising models ($\beta=\gamma<1$) in chapter 3 of his Ph.D. thesis \cite{liuthesis}, where similar zero-freeness results (a complex neighborhood of $\mathcal{S}^{\Delta}_2$ restricted to $\beta=\gamma$) were obtained. 
We mention that
by using the unique analytic continuation and the inverse function theorem, our main technical result (Theorem \ref{real-to-complex}) is generic; it does not rely on a particularly chosen potential function. 
Thus, in our approach we can work with \emph{any} existing potential function based arguments for correlation
decay even if the potential function does not have an explicit expression, for instance, the one used in \cite{lly13} when $\beta\neq \gamma$.
Furthermore, we mention also that based on the zero-freeness, we obtain new correlation decay result for  complex parameters (Lemma \ref{lem:complex-to-correlation}).
Note that Barvinok's algorithm requires an entire region in which the partition function is zero-free and there is an easy computing point.
However, our correlation decay result shows that one can always devise an FPTAS for these parameter settings via Weitz's algorithm, even if Barvinok's algorithm fails.



\section{Weitz's Algorithm}\label{sec2}

In this section, we describe Weitz's algorithm. 
We first consider 
positive parameters $\pmb \zeta  \in \mathbb{R}^3_{+}$.
An obvious but important fact about $\pmb \zeta$ being positive
is that $Z_G(\pmb \zeta)\neq 0$ for any graph $G$. 
This is true even if $G$ contains arbitrary number of  vertices pinned to spin $+$ or $-$.
Then, the  partition  function  can  be  viewed  as  the  normalizing  factor  of  the  Gibbs  distribution.
\subsection{Notations and definitions}
Let $\pmb \zeta \in \mathbb{R}^3_{+}$. We use $p_v(\pmb \zeta)$ to denote the marginal probability of $v$ being assigned to spin $+$ in the Gibbs distribution, i.e., $p_v(\pmb \zeta)=\frac{Z_{G,v}^+(\pmb \zeta)}{Z_G(\pmb \zeta)}$,
where $Z_{G,v}^+(\pmb \zeta)$ is the contribution to $Z_G(\pmb \zeta)$ 
over all configurations with  $v$ being assigned to spin $+$.
We know that $p_v$ is well-defined since $Z_G(\pmb \zeta)\neq 0$.

Let $\sigma_{\Lambda}\in\{0,1\}^\Lambda$ be a configuration of some subset $\Lambda\subseteq V$. We allow $\Lambda$ to be the empty set.
We call vertices in $\Lambda$ \emph{pinned} and other vertices \emph{free}.
We use $p_v^{\sigma_{\Lambda}}(\pmb \zeta)$ to denote the marginal probability of a free vertex $v$ ($v\notin\Lambda$) being assigned to spin $+$
conditioning on the configuration $\sigma_{\Lambda}$ of $\Lambda$, i.e., 
$p^{\sigma_{\Lambda}}_v(\pmb \zeta)=\frac{Z_{G,v}^{\sigma_{\Lambda}, +}(\pmb \zeta)}{Z_G^{\sigma_{\Lambda}}(\pmb \zeta)}$, where $Z_G^{\sigma_{\Lambda}}(\pmb \zeta)$ is the weight over all configurations where vertices in $\Lambda$ are pinned by the configuration $\sigma_{\Lambda}$, and $Z_{G, v}^{\sigma_{\Lambda}, +}(\pmb \zeta)$ is the contribution to $Z_G^{\sigma_{\Lambda}}(\pmb \zeta)$ with $v$ being assigned to spin $+$.
Correspondingly, we can define $Z_{G, v}^{\sigma_{\Lambda}, -}(\pmb \zeta)$.
Let $R_{G,v}^{\sigma_\Lambda}(\pmb \zeta):=\frac{Z_{G, v}^{\sigma_{\Lambda}, +}(\pmb \zeta)}{Z_{G, v}^{\sigma_{\Lambda}, -}(\pmb \zeta)}=\frac{p_v^{\sigma_\Lambda}(\pmb \zeta)}{1-p_v^{\sigma_\Lambda}(\pmb \zeta)}$ be the ratio between the two probabilities
that the free vertex $v$ is assigned to spin $+$ and $-$, while imposing some condition $\sigma_\Lambda.$
Since $Z_G(\pmb \zeta)\neq 0$ for any graph $G$ with arbitrary number of pinned vertices,
both $p_v^{\sigma_\Lambda}(\pmb \zeta)$ and $R_{G,v}^{\sigma_\Lambda}(\pmb \zeta)$ are well-defined. 
When context is clear, we write $p_v(\pmb \zeta)$ $p^{\sigma_{\Lambda}}_v(\pmb \zeta)$ and $R_{G,v}^{\sigma_\Lambda}(\pmb \zeta)$ as $p_v$, $p^{\sigma_{\Lambda}}_v$ and $R_{G,v}^{\sigma_\Lambda}$ for convenience.

Since computing the partition function of 2-spin systems is self-reducible, 
if one can compute $p_v$ for any vertex $v$, then the partition function can be computed via telescoping \cite{jvv86}. 
The goal of Weitz's algorithm is to estimate $p_v^{\sigma_\Lambda}$, which is equivalent to estimating $R_{G,v}^{\sigma_\Lambda}$.
For the case that the graph is a tree $T$,
$R_{T,v}^{\sigma_\Lambda}$ can be computed by recursion.
Suppose that a free vertex $v$ has $d$ children, and $s_1$ of them are pinned to  $+$, $s_2$ are pinned to  $-$, and $k$ are free $(s_1+s_2+k=d)$. 
We denote these $k$ free vertices by $v_i (i\in [k])$ and let $T_i$ be the corresponding subtree rooted at $v_i$. 
We use $\sigma^i_{\Lambda}$ to denote the configuration $\sigma_{\Lambda}$ restricted to $T_i$.
Since all subtrees are independent, it is easy to get the following recurrence relation,
\begin{equation*}
\begin{aligned}
R^{\sigma_\Lambda}_{T,v}=\frac{Z_{T,v}^{\sigma_\Lambda,+}(\pmb \zeta)}{Z_{T,v}^{\sigma_\Lambda,-}(\pmb \zeta)}=\frac{\lambda^{1+s_1}\beta^{s_1}\prod_{i=1}^{k}{\left(\beta Z_{T_i,v_i}^{\sigma_\Lambda^i,+}(\pmb \zeta)+Z_{T_i,v_i}^{\sigma_\Lambda^i,-}(\pmb \zeta)\right)}}{\lambda^{s_1}\gamma^{s_2}\prod_{i=1}^{k}{\left(Z_{T_i,v_i}^{\sigma_\Lambda^i,+}(\pmb \zeta)+\gamma Z_{T_i,v_i}^{\sigma_\Lambda^i,-}(\pmb \zeta)\right)}}
=\lambda\beta^{s_1}\gamma^{-s_2}\prod_{i=1}^{k}{\left(\frac{\beta R^{\sigma_\Lambda^i}_{T_i,v_i}+1}{R^{\sigma_\Lambda^i}_{T_i,v_i}+\gamma}\right)}.\\
\end{aligned}
\end{equation*}

\begin{definition}[Recursion function]
Let $\mathbf s=(s_1, s_2, k)\in \mathbb{N}^3$ {(}including 0{)}. A recursion function $F_{\mathbf s}$ for 2-spin systems is defined to be 
$$F_{\mathbf s}(\pmb \zeta, \mathbf{x}):=\lambda\beta^{s_1}\gamma^{-s_2}\prod_{i=1}^{k}{\left(\frac{\beta x_i+1}{x_i+\gamma}\right)},$$
where $\pmb \zeta=(\beta, \gamma, \lambda)\in \mathbb{C}\times(\mathbb{C}\backslash\{0\})\times\mathbb{C}$ and $\mathbf{x}=(x_1, \ldots, x_k)\in (\mathbb{C}\backslash\{-\gamma\})^k$.
We define $F_{\pmb \zeta, \mathbf s}(\mathbf{x}):=F_{\mathbf s}(\pmb \zeta, \mathbf{x})$ for fixed $\pmb \zeta$ with $\gamma\neq 0$, and $F_{\mathbf x, \mathbf s}(\pmb \zeta):=F_{\mathbf s}(\pmb \zeta, \mathbf{x})$ for fixed $\mathbf{x}$.
\end{definition}
\begin{remark}
Every recursion function is analytic on its domain. 
\end{remark}

For a general graph $G$, 
Weitz  reduced computing $R_{G,v}^{\sigma_\Lambda}$ 
to
that in a tree $T$, called the self-avoiding walk (SAW) tree, and Weitz's theorem states that $R_{G,v}^{\sigma_\Lambda}=R_{T,v}^{\sigma_\Lambda}$ \cite{weitz}. (See the appendix for more details.)
We want to generalize Weitz's theorem to complex parameters $\pmb \zeta \in \mathbb{C}^3$.
First, we need to make sure that $R_{G,v}^{\sigma_\Lambda}$ and $p_v^{\sigma_\Lambda}$ are
well-defined for vertex $v\notin\Lambda$. 
This requires that  $Z^{\sigma_\Lambda}_{G}(\pmb \zeta)\neq 0$ for any graph $G$ and any configuration $\sigma_\Lambda$.
Now, $p_v^{\sigma_\Lambda}$  no longer has a probabilistic meaning. It is just a ratio of two complex numbers.
However, one can easily observe that for some special parameters, 
there are trivial configurations such that
 $Z^{\sigma_\Lambda}_{G, v}(\pmb \zeta)=0$.
We will rule these cases out as they are \emph{infeasible}.

\begin{definition}[Feasible configuration]\label{def:feasible-configuration}
Let $\pmb \zeta \in \mathbb{C}^3$. Given a  graph $G=(V, E)$ of the 2-spin system specified by  $\pmb \zeta$, 
a configuration $\sigma_\Lambda$ on some vertices $\Lambda \subseteq V$  is feasible if 
\begin{itemize}
    \item  $\sigma_\Lambda$ does not assign any vertex in $G$ to spin $+$ if $\lambda=0$, and
    \item  $\sigma_\Lambda$ does not assign any two adjacent vertices in $G$ both to spin $+$ if $\beta=0$.
\end{itemize}
\end{definition}
\begin{remark}
Let  $\sigma_\Lambda$ be a feasible configuration.
If we further pin one vertex $v\notin \Lambda$ to spin $-$, 
and get the configuration $\sigma_{\Lambda'}$ on $\Lambda'=\Lambda\cup\{v\}$, 
then $\sigma_{\Lambda'}$ is still a feasible configuration. 
Thus, given $\pmb \zeta \in \mathbb{C}^3$, if $Z^{\sigma_\Lambda}_{G}(\pmb \zeta)\neq 0$ for any graph $G$ and any arbitrary feasible configuration $\sigma_\Lambda$ on $G$, then both  $p_v^{\sigma_\Lambda}$  and $R_{G,v}^{\sigma_\Lambda}$ are well-defined.
\end{remark}

Given $R_{G,v}^{\sigma_\Lambda}$ is well-defined for some $\pmb \zeta \in \mathbb{C}^3$, 
we can still compute it by recursion via SAW tree.
We first consider the case that $\lambda \neq 0$.
Let $\sigma_\Lambda$ be a feasible configuration.
It is easy to verify that the corresponding configuration on the SAW tree is also   feasible and Weitz's theorem still holds.
For the case that $\lambda=0$, it is obvious that $R_{G,v}^{\sigma_\Lambda}\equiv 0$ for any graph $G$, any free vertex $v$ and any feasible configuration $\sigma_\Lambda$.
This is equal to the value of recursion functions $F_{\mathbf{s}}(\pmb{\zeta}, \mathbf{x})$ at $\lambda=0$.
We agree that $R_{G,v}^{\sigma_\Lambda}$ can be computed by recursion functions when $\lambda=0$, although Weitz's theorem does not hold for this case.
For the case that $\beta =0$, we have $R_{G,v}^{\sigma_\Lambda}= 0$ if one of the children of $v$ is pinned to $+$. 
Then, we may view $v$ as it is pinned to $-$.
Thus, for $\beta=0$, we only consider recursion functions $F_{\mathbf{s}}$ where $s_1=0$.

\subsection{Correlation decay and real contraction}
The SAW tree may be exponentially large in size of $G$.
In order to get a polynomial time approximation algorithm, we may run the tree recursion at logarithmic depth and hence in polynomial time, and plug in some arbitrary values at the truncated boundary. 
We have the following notion of \emph{strong spatial mixing} (SSM) to bound the error caused by arbitrary guesses.
It was originally introduced for non-negative parameters.
Here, we extend it to complex parameters. 

\begin{definition}[Strong spatial mixing]\label{def:correlation-decay}
  A 2-spin system specified by $\pmb \zeta\in \mathbb{C}^3$ on a family $\mathcal G$ of graphs is said to exhibit strong spatial mixing if for any graph $G=(V,E)\in\mathcal{G}$,
  any $v\in V$, and any feasible configurations $\sigma_{\Lambda_1}\in\{0,1\}^{\Lambda_1}$ and $\tau_{\Lambda_2}\in\{0,1\}^{\Lambda_2}$ where $v \notin \Lambda_1 \cup \Lambda_2$, we have
  \begin{enumerate}
      \item $Z_{G}^{\sigma_{\Lambda_1}}(\pmb \zeta)\neq 0$ and $Z_{G}^{\tau_{\Lambda_2}}(\pmb \zeta)\neq 0$, and
      \item $\big|{p_v^{\sigma_{\Lambda_1}}-p_v^{\tau_{\Lambda_2}}}\big|\leq \exp(-\Omega(\mathrm{dist}(v,S)))$,
  \end{enumerate}
  where 
  $S\subseteq\Lambda_1\cup\Lambda_2$ is the subset on which $\sigma_{\Lambda_1}$ and $\tau_{\Lambda_2}$ differ\footnote{If a vertex $v$ is free in one configuration but pinned in the other, we say these two configurations differ at $v$.},
  and 
  $\mathrm{dist}_{G}(v, S)$ is the shortest distance from $v$ to any vertex in $S$.
\end{definition}
\begin{remark}
When $\pmb \zeta\in \mathbb{R}_{+}^3$, condition 1 is always satisfied. 
Condition 2 is a stronger form of SSM of real parameters (see Definition 5 of \cite{lly13}).
For real values, by monotonicity one can restrict to the case that $\Lambda_1=\Lambda_2$ (the two configurations are on the same set of vertices). 
Here, we allow $\Lambda_1\neq \Lambda_2$.
\end{remark}
In statistical physics, SSM is called correlation decay. If SSM holds,
 then the error caused by arbitrary boundary guesses at logarithmic depth of the SAW tree is polynomially small.
Hence, Weitz's algorithm gives an FPTAS.
A main technique that has been widely used to establish SSM
is the \emph{potential method} \cite{rstvy, lly12, lly13, sssy17, gl18}.
Instead of bounding the rate of decay of recursion functions directly, 
we use a potential function $\varphi(x)$ to map the original recursion to a new domain (See Figure \ref{fig:commutatve-diagram} for the commutative diagram).
	\begin{figure}[!hbtp]
\centering
		\includegraphics[scale=0.4]{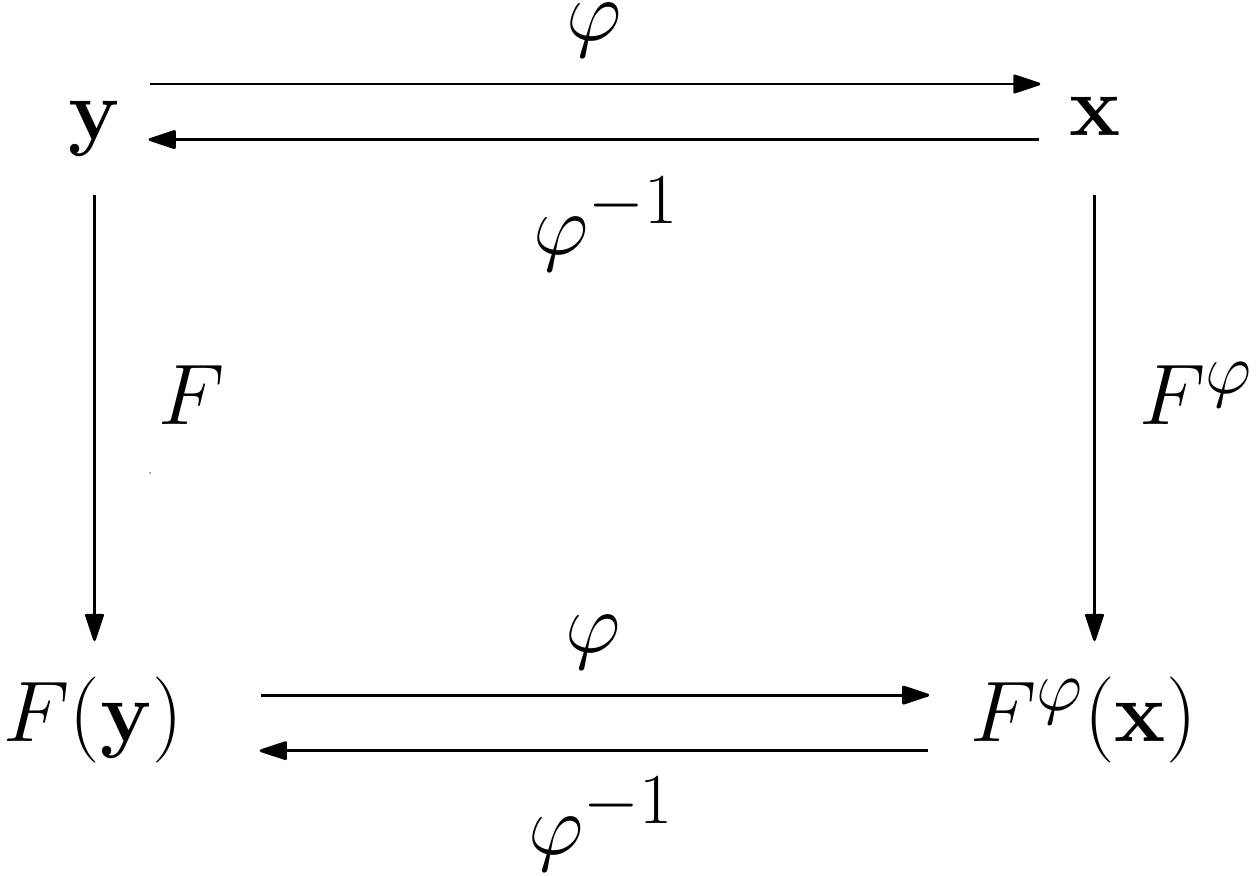}
		\vspace{-1ex}
	\caption{Commutative diagram between $F$ and $F^{\varphi}$}
	\label{fig:commutatve-diagram}
	\end{figure}
Let $F_{\mathbf s}(\pmb \zeta, \mathbf{y})$ be a recursion function.
We use  $F^{\varphi}_{\mathbf s}(\pmb \zeta, \mathbf{x})$
to denote the composition 
$\varphi(F_{\mathbf s}(\pmb \zeta, {\pmb \varphi^{-1}(\mathbf x)}))$
where $\mathbf{y}={\pmb \varphi^{-1}(\mathbf x)}$ denotes the vector $(\varphi^{-1}(x_1), \ldots, \varphi^{-1}(x_k))$.
Correspondingly, we define $F^\varphi_{\pmb{\zeta}, \mathbf s}(\mathbf{x})$ for fixed $\pmb{\zeta}$, and $F^\varphi_{\mathbf{x}, \mathbf s}(\pmb{\zeta})$ for fixed $\mathbf{x}$.
We will specify the domain on which $F^\varphi_{\mathbf s}$ is well-defined 
per each $\varphi$ that will be used.
For positive $\pmb{\zeta}$,
a sufficient condition for the bounded 2-spin system of $\pmb{\zeta}$  exhibiting SSM 
is that there exists a ``good'' potential function $\varphi$ such that $F^\varphi_{\pmb \zeta, \mathbf s}$ satisfies the following contraction property. 
\begin{definition}[Real contraction]\label{def:real-contract}
Fix $\Delta \in \mathbb{N}$.
We say $\pmb{\zeta}\in \mathbb{R}^3$ satisfies  real contraction  for $\Delta$
if there is 
 a
real compact  interval $J\subseteq \mathbb{R}$ where $\lambda\in J$, $-\gamma\notin J$ and $-1\notin J$,
and a real analytic function $\varphi: J \rightarrow I$ where $\varphi'(x)\neq 0$ for all $x\in J$,
such that 
\begin{enumerate}
 \item  $F_{\pmb \zeta, \bf s}(J^{k})\subseteq J$ for every $\mathbf{s}$ with $\|\mathbf{s}\|_1\leq \Delta -1$ and $-1\notin F_{\pmb \zeta, \bf s}(J^{k})$ for every $\mathbf{s}$ with $\|\mathbf{s}\|_1=\Delta$;
    \item there exists $\eta>0$ s.t. $\left\|\nabla F^\varphi_{\pmb \zeta, \bf s}(\mathbf{x})\right\|_1\le1-\eta$  for every $\mathbf{s}$ with $\|\mathbf{s}\|_1\leq \Delta -1$ and  all $\mathbf{x} \in I^{k}$.
\end{enumerate}
We say $\varphi$ defined on $J$ is a good potential function for $\pmb{\zeta}$.
\end{definition}

\begin{remark} 
Since $\varphi$ is analytic and $\varphi'(x)\neq 0$ for $x\in J$, 
we have $\varphi$ is invertible and the inverse $\varphi^{-1}:I\rightarrow J$ is also analytic by the inverse function theorem.
Also for every $\mathbf s$ with $\|{\mathbf s}\|_1\leq \Delta-1$, since $F_{\pmb \zeta, \bf s}(J^{k})\subseteq J$ and $F_{\pmb \zeta, \bf s}(\mathbf{x})$ is analytic on $J^k$ due to $-\gamma\notin J$, 
we have  $F^\varphi_{\pmb \zeta, \bf s}(\mathbf{x})$ is well-defined and analytic on $I^k$. Then $\nabla F^\varphi_{\pmb \zeta, \bf s}(\mathbf{x})$ is well-defined on $I^k$. 
We know $I$ is also a  real compact interval since  $J$ is a real compact interval and $\varphi$ is a real analytic function.

\end{remark}

Note that since $-1\notin J$, $F_{\pmb \zeta, \bf s}(J^{k})\subseteq J$ implies that $-1\notin F_{\pmb \zeta, \bf s}(J^{k})$.
Thus, real contraction implies that $-1\notin F_{\pmb \zeta, \bf s}(J^{k})$ for all $\|\mathbf{s}\|_1\leq \Delta$.
The reason why we require $F_{\pmb \zeta, \bf s}(J^{k})\subseteq J$ for $\|\mathbf{s}\|_1\leq \Delta-1$, but only require $-1\notin F_{\pmb \zeta, \bf s}(J^{k})$ for $\|\mathbf{s}\|_1= \Delta$
is that in a tree of degree at most $\Delta$,
 only the root node may have $\Delta$ many children,
 while other nodes have at most $\Delta-1$ many children.

\begin{lemma}\label{lem:real-correlation}
If $\pmb{\zeta}\in \mathbb{R}^3_{+}$ satisfies real contraction for $\Delta$, then the $\Delta$-bounded 2-spin system of $\pmb{\zeta}$
exhibits SSM, and hence there is an FPTAS for computing the partition function $Z_G(\pmb \zeta)$.
\end{lemma}
\vspace{-1ex}
\begin{proof}
The proof directly follows from the argument of the potential method, see \cite{lly13, gl18}. 
The FPTAS follows from Weitz's algorithm.
\end{proof}
\vspace{-1ex}
Now, we give the sets of non-negative parameters which satisfy real contraction. 

\begin{definition}[Uniqueness condition \cite{lly13}]\label{def:uniqueness}
Let $\pmb \zeta\in \mathbb{R}^3$ be anti-ferromagnetic ($\beta\gamma<1$) with $\beta \geq 0$, $\gamma>0$ and $\lambda\geq 0$, and $f_d(x)=\lambda\left(\frac{\beta x+1}{x+\gamma}\right)^d$.
We say $\pmb \zeta$ is up-to-$\Delta$ unique, if $\lambda=0$ or $\lambda>0$ and there exists a constant $0<c<1$ such that for every integer $1\le d\le\Delta-1$,
$$\left|f'_d(\hat{x}_d)\right|=\frac{d(1-\beta\gamma)\hat{x}_d}{(\beta\hat{x}_d+1)(\hat{x}_d+\gamma)}\leq c,$$ where $\hat{x}_d$ is the unique positive fixed point of the function $f_d(x)$. 
\end{definition}

Let $\mathcal{S}^{\Delta}_i (i\in [4])$ be the correlation decay sets defined in Definition \ref{Correaltion-decay-sets}.
The set $\mathcal{S}^{\Delta}_1$ was given in \cite{zhangbai} and $\mathcal{S}^{\Delta}_2$ was given in \cite{lly13}. 
Directly following their proofs, it is easy to verify that both sets satisfy real contraction.
The sets $\mathcal{S}^{\Delta}_3$ and $\mathcal{S}^{\Delta}_4$  are obtained in this paper, and we show that they also satisfy real contraction.
We will give a proof in the appendix for every $\mathcal{S}^{\Delta}_i (i\in [4])$.

\begin{lemma}\label{lem:c1c2}
Fix $\Delta \geq 3$. For every $\pmb \zeta \in \mathcal{S}^{\Delta}_i (i\in [4])$, it satisfies real contraction for $\Delta$.
\end{lemma}

In order to generalize the correlation decay technique to complex parameters, 
we need to ensure that the partition function is zero-free.
Now, let us first take a detour to Barvinok's algorithm 
which crucially relies on the zero-free regions of the partition function.
After we carve out our new zero-free regions, we will come back to the existence of correlation decay of complex parameters.

\section{Barvinok's Algorithm}\label{sec3}
In this section, we describe Barvinok's algorithm.
Let $I=[0, t]$ be a closed real interval.
We define  the $\delta$-strip of $I$ to be $\{z\in \mathbb{C}\mid |z-z_0|<\delta, z_0\in I\}$, denoted by $I_{\delta}$. 
It is a complex neighborhood of $I$.
Suppose a graph polynomial $P(z)=\sum^n_{i=0}a_iz^i$ of degree $n$ is zero-free in $I_{\delta}$.
Barvinok's method \cite{bar16} roughly states that for any $z\in I_{\delta}$,
$P(z)$ can be $(1\pm\varepsilon)$-approximated using coefficients
$a_0, \ldots, a_k$ for some $k=O(e^{\Theta(1/\delta)}\log(n/\varepsilon))$, via truncating the Taylor expansion of the logarithm of the polynomial.
For the partition function of 2-spin systems,
these coefficients can be computed in polynomial-time \cite{pr17, lss19b}.  
For the purpose of obtaining FPTAS, we will view 
the partition function
 as a univariate polynomial $Z_{G; \beta, \gamma}(\lambda)$ in $\lambda$ and fix $\beta$ and $\gamma$.
The following result is known. 

\begin{lemma}\label{lem-bar}
Fix  $\beta, \gamma \in \mathbb{C}$ and $\Delta \in \mathbb{N}$.
Let $G$ be a graph of degree at most $\Delta$.
If $Z_{G; \beta, \gamma}(\lambda)\neq 0$ lies in a $\delta$-strip $I_{\delta}$ of $I=[0, t]$, then there is an FPTAS for  computing $Z_{G; \beta, \gamma}(\lambda)$ for $\lambda \in I_{\delta}$.
\end{lemma}
\begin{proof}
\vspace{-1ex}
This lemma is a generalization of Lemma 4 in \cite{gll19}, where $\beta$ and $\gamma$ are both real.
The generalization to complex valued parameters directly follows from the argument in \cite{lss19b}. 
\vspace{-1ex}
\end{proof}

\subsection{Zero-freeness and complex contraction}
With Lemma \ref{lem-bar} in hand, the main effort is to obtain zero-free regions  of the partition function.
For this purpose, we will still view 
$Z_G(\pmb \zeta)$ 
as a multivariate polynomial in $(\beta, \gamma, \lambda)$.
A main and widely-used approach to obtain zero-free regions is the \emph{recursion} method \cite{sok01, ss05, bc18, pr19, lss19a}. This method is related to the correlation decay method.

Assuming  $Z^{-}_{G,v}(\pmb \zeta)\ne0$ for some vertex $v$, then $Z_G(\pmb \zeta)\ne0$ is equivalent to $R_{G,v}=\frac{Z^{+}_{G,v}(\pmb \zeta)}{Z^{-}_{G,v}(\pmb \zeta)}\ne-1$.
As pointed above, the ratio $R_{G,v}$ can be computed by recursion via the SAW tree in which $v$ is the root.
Roughly speaking, the key idea of the recursion method is to construct a \emph{contraction} region $Q\subseteq\mathbb{C}$ 
where $\lambda\in Q$ and $-1\notin Q$ 
such that for all recursion functions $F_{\pmb{\zeta}, \mathbf s}$ with $\|\mathbf s\|_1\leq \Delta-1$,
$F_{\pmb{\zeta}, \mathbf s}(Q^k)\subseteq Q$ and for all $F_{\pmb{\zeta}, \mathbf s}$ with $\|\mathbf s\|_1= \Delta$,
$-1 \notin F_{\pmb{\zeta}, \mathbf s}(Q^k)$.
This condition guarantees that with the initial value $R_{G,v_\ell}=\lambda$ where $v_\ell$ is a free leaf node in the SAW tree of which the degree is bounded by $\Delta$, the recursion will never achieve $-1$.
Hence, we have $Z_G(\pmb \zeta)\ne0$ by induction.
Again, we may use a potential function $\varphi: Q\rightarrow P$ to change the domain, and we prove $F^{\varphi}_{\pmb{\zeta}, \mathbf s}(P^k)\subseteq P$.

Now, we introduce the following \emph{complex} contraction property as a generalization of real contraction. This property gives a sufficient condition for the zero-freeness of the partition function.

\begin{definition}[Complex contraction]\label{complex contraction}
Fix $\Delta \in \mathbb{N}$.
We say $\pmb{\zeta}\in \mathbb{C}^3$ satisfies complex contraction for $\Delta$
if there is  a closed and bounded complex region   $Q\subseteq \mathbb{C}$ where $\lambda \in Q$, $-\gamma \notin Q$ and $-1\notin Q$,
and an analytic and invertible function $\varphi: Q \rightarrow P$ where the inverse $\varphi^{-1}: P\rightarrow Q$ is also analytic and $P$ is convex,
such that 
\begin{enumerate}
 \item  $F_{\pmb \zeta, \bf s}(Q^{k})\subseteq Q$ for every $\mathbf{s}$ with $\|\mathbf{s}\|_1\leq \Delta -1$  and $-1\notin F_{\pmb \zeta, \bf s}(Q^{k})$ for every $\mathbf{s}$ with $\|\mathbf{s}\|_1=\Delta$;
    \item there exists $\eta>0$ s.t. $\left\|\nabla F^\varphi_{\pmb \zeta, \bf s}(\mathbf{x})\right\|_1\le1-\eta$  for every $\mathbf{s}$ with $\|\mathbf{s}\|_1\leq \Delta -1$ and  all $\mathbf{x} \in P^{k}$.
\end{enumerate}
\end{definition}
\begin{remark}
Similar to the remark of Definition \ref{def:real-contract}, we have $F^\varphi_{\pmb \zeta, \bf s}(\mathbf{x})$ is well-defined and analytic on $P^k$. Here, we directly assume that the inverse $\varphi^{-1}$ is analytic instead of $\varphi'(x)\neq 0$ for the sake of simplicity of our proof. 
\end{remark}

\begin{lemma}\label{complex-to-zerofree}
If $\pmb \zeta$ satisfies complex contraction for $\Delta$, then $Z_{G}^{\sigma_{\Lambda}}(\pmb \zeta)\neq 0$ for any graph $G$ of degree at most $\Delta$ and any feasible configuration $\sigma_{\Lambda}$.
\end{lemma}

Please see the appendix for the proof. 
Such a proof only uses  condition 1 of complex contraction.
However, condition 2 combining with the zero-freeness result of Lemma \ref{complex-to-zerofree} 
gives a sufficient condition for bounded 2-spin systems of \emph{complex} parameters exhibiting correlation decay. 
This is a generalization of Lemma \ref{lem:real-correlation}. 
Also, we will give the proof in the appendix.

\begin{lemma}\label{lem:complex-to-correlation}
If $\pmb{\zeta}$ satisfies complex contraction for $\Delta$, then the $\Delta$-bounded 2-spin system specified by $\pmb{\zeta}$
exhibits correlation decay. Thus, there is an FPTAS for computing $Z_{G}(\pmb\zeta)$ via Weitz's algorithm.
\end{lemma}

\section{From Real Contraction to Complex Contraction}\label{sec4}
In this section, we will prove our main result. 
We first give some preliminaries in complex analysis.
The main tools are the unique analytic continuation and the inverse function theorem.
Here, we slightly modify the statements to fit for our settings.
Please refer to \cite{complex} for the proofs.

\begin{theorem}[Unique analytic continuation]\label{analytic-continue}
Let $f({x})$ be a (real) analytic function defined on a compact real interval $I \subseteq \mathbb{R}$.
Then, there exists a complex 
neighborhood  $\widetilde{I}\subseteq \mathbb{C}$ of ${I}$, and a (complex) analytic function $\widetilde{f}({x})$ defined on   $\widetilde{ I}$ such that 
$\widetilde{f}({x})\equiv{f}({x})$ for all ${x} \in {I}$.
Moreover, if there is another (complex) analytic function $\widetilde{g}({x})$ also defined on  $\widetilde{I}$ such that 
$\widetilde{g}({x})\equiv\widetilde{f}({x})$ for all ${x} \in {I}$ and the measure $\frak m({I})\neq 0$,
then 
$\widetilde{g}({x})\equiv\widetilde{f}({x})$ for all ${x} \in \widetilde{{I}}$.
We call $\widetilde{f}({x})$ the unique analytic continuation of ${f}({x})$ on $\widetilde{{I}}$.
\end{theorem}



\begin{theorem}[Inverse function theorem]\label{inverse-function}
Let $\varphi$ be a  (complex) analytic function defined on $U\subseteq \mathbb{C}$, and $\varphi'(z)\neq 0$ for some $z\in U$.
Then there exists a complex neighborhood $D$ of $z$ such that $\varphi$ is invertible on $D$ and the inverse is also analytic.
\end{theorem}

Combining the above theorems, we have the following result.

\begin{lemma}\label{lem-invertible}
Let $\varphi:J\rightarrow I$ be a real analytic function, and $\varphi'(x)\neq 0$ for all $x\in J$ where $J$ and $I$ are both real compact intervals.
Then, there exists an analytic continuation $\widetilde{\varphi}$ on  a complex neighborhood $\widetilde{J}$ of $J$ such that $\widetilde{\varphi}$ is invertible on 
$\widetilde{J}$ and the inverse $\widetilde{\varphi}^{-1}$ is also analytic.
\end{lemma}
\begin{proof}
If $\frak{m}(J)=0$, i.e., $J=\{x\}$, then by Theorem \ref{inverse-function}, there exists an analytic continuation $\widetilde{\varphi}$ of $\varphi$ defined on a neighborhood of $x$ on which $\widetilde{\varphi}$ is invertible and the inverse $\widetilde{\varphi}^{-1}$ is analytic.

Otherwise, $\frak{m}(J)\neq0$.
Since $\varphi(x)$ is analytic and $\varphi'(x)\neq 0$ for all $x \in J$,
 we have $\varphi$ is invertible and by Theorem \ref{inverse-function}, the inverse
$\varphi^{-1}: I\rightarrow J$ is analytic on $I$.
By Theorem \ref{analytic-continue}, there exists an analytic continuation $\widetilde{\varphi^{-1}}$ of $\varphi^{-1}$ defined on a neighborhood  $\widetilde{I}_1$ of $I$.
Similarly, there exists an analytic continuation $\widetilde{\varphi}$ of $\varphi$ defined on a neighborhood $\widetilde{J}$ of $J$.
We use $\widetilde{I}$ to denote the image $\widetilde{\varphi}(\widetilde{J})$. 
Since $\widetilde{\varphi}$ is analytic and by the open mapping theorem, we know $\widetilde{I}$ is an open set in the complex plane.
Clearly, we have $\varphi(J)=I\subseteq \widetilde{I}$.
We can pick $\widetilde{J}$ small enough while still keeping $J \subseteq \widetilde{J}$ such that the image $\widetilde{I}=\widetilde{\varphi}(\widetilde{J})\subseteq \widetilde{I}_1$ and still $I \subseteq \widetilde{I}$.
Thus, the composition $\widetilde{\varphi^{-1}}\circ \widetilde{\varphi}$ is a well-defined analytic function on $\widetilde{J}$.
Clearly, we have 
$$\widetilde{\varphi^{-1}}\circ \widetilde{\varphi}(x)={\varphi^{-1}}\circ {\varphi}(x)\equiv x \text{ for all $x\in J$}.$$
Since $\frak{m}(J)\neq0$, by Theorem \ref{analytic-continue},
we have $\widetilde{\varphi^{-1}}\circ \widetilde{\varphi}(x)\equiv x \text{ for all $x\in \widetilde{J}$}.$

Thus, $\widetilde{\varphi}$ is invertible on $\widetilde{J}$ and the inverse $\widetilde{\varphi}^{-1}=\widetilde{\varphi^{-1}}$ is analytic.
\end{proof}

Now, we are ready to prove our main result.
\begin{theorem}\label{real-to-complex}
If $\pmb \zeta_0$ satisfies  real contraction for $\Delta$, then there exists a $\delta>0$ such that for every $\pmb \zeta\in \mathbb{C}^3$ with $\|\pmb \zeta-\pmb \zeta_0\|_{\infty}<\delta$, $\pmb \zeta $ satisfies  complex contraction for $\Delta$.

\end{theorem}
\begin{proof}
Let $\varphi:J \rightarrow I$ be a good potential function for $\pmb \zeta_0$. 
By Definition \ref{def:real-contract} and Lemma \ref{lem-invertible}, there exists a neighborhood $\widetilde{J}$ of $J$ such that 
the analytic continuation $\widetilde{\varphi}:\widetilde{J}\rightarrow \widetilde{I}$ of $\varphi$ on $\widetilde{J}$ is invertible.
Here $\widetilde{I}=\widetilde{\varphi}(\widetilde{J})$ is a neighborhood of $I$, and the inverse $\widetilde{\varphi}^{-1}$ is also analytic on $\widetilde{I}$.
We use $\mathcal{B}_{\delta}:=\{\mathbf{z}\in\mathbb{C}^3\mid \|\mathbf{z}-\pmb \zeta_0\|_{\infty}<\delta\}$ to denote the 3-dimensional complex ball of radius $\delta$ in terms of infinity norm around $\pmb \zeta_0$.
Recall that we define $I_{\varepsilon}=\{z\in \mathbb{C}\mid |z-z_0|<\varepsilon, z_0\in I\}$.
Given a set $U\subseteq \mathbb{C}^k$, we use $\overline{U}$
to denote its closure.

We first show that we can pick a pair of $(\delta_1, \varepsilon_1)$ 
such that for every $\mathbf s$ with  $\|\mathbf s\|_1\leq \Delta-1$, the composition $$F^{\widetilde{\varphi}}_{\mathbf s}(\pmb \zeta, \mathbf{x})=\widetilde\varphi(F_{\mathbf s}(\pmb \zeta, { \widetilde{\pmb\varphi}^{-1}(\mathbf x)})) \text{ is well-defined and analytic on }
\mathcal{B}_{\delta_1}\times I_{\varepsilon_1}^k.$$
Given some $\mathbf s$ with $\|\mathbf s\|_1\leq \Delta-1$, 
we consider the function $F_{\mathbf s}(\pmb \zeta, \mathbf{x})$.
We know that it is analytic on a neighborhood of $\{\pmb\zeta_0\}\times J^k$ 
and  by real contraction we have $F_{\mathbf s}(\pmb \zeta_0, J^k)\subseteq J$.
Then, we can pick some $\delta_{\mathbf s}$ and a neighborhood $\widetilde{J}_{\mathbf s}$ of $J$ that are small enough such that 
$F_{\mathbf s}(\pmb \zeta, \mathbf{x})$ is analytic on $\mathcal{B}_{\delta_{\mathbf s}}\times \widetilde{J}_{\mathbf s}^k$, 
and  $F_{\mathbf s}(\mathcal{B}_{\delta_{\mathbf s}}, \widetilde{J}_{\mathbf s}^k)\subseteq \widetilde{J}$.
Let $$\delta_1=\min_{{\|\mathbf s\|_1}\leq\Delta-1}\{\delta_{\mathbf s}\}\text{~~~ and ~~~}
\widetilde{J}_1=\bigcap_{{\|\mathbf s\|_1}\leq\Delta-1}\widetilde{J}_{\mathbf s}.$$
Since there is only a finite number of $\mathbf s$ with $\|\mathbf s\|_1\le \Delta -1$, we know $\delta_1>0$, and  $\widetilde{J}_1$ is open and it is a neighborhood of $J$. 
We have $F_{\mathbf s}(\mathcal{B}_{\delta_1}, \widetilde{J}_1)\subseteq \widetilde{J}$ for every $\mathbf s$ with ${\|\mathbf s\|_1}\leq\Delta-1$.
Since $\widetilde\varphi^{-1}$ is analytic on $\widetilde{I}$ and  $\widetilde\varphi^{-1}(I)=J$, 
similarly we can pick a small enough neighborhood $\widetilde{I}_1$ of $I$  
where $\widetilde{I}_1\subseteq\widetilde{I}$ such that
$\widetilde\varphi^{-1}(\widetilde{I}_1)\subseteq \widetilde{J}_1$.
For every $z_0\in I$, we can pick an $\varepsilon_{z_0}$ such that the disc $B_{z_0, \varepsilon_{z_0}}:=\{z\in\mathbb{C}\mid |z-z_0|<\varepsilon_{z_0}\}$ is in $\widetilde{I}_1$. 
Recall that $I$ is a compact real interval, 
by the finite cover theorem, 
we can uniformly pick a $\varepsilon_1$ such that $I\subseteq I_{\varepsilon_1}\subseteq \widetilde{I}_1$. 
Thus, we have $F^{\widetilde{\varphi}}_{\mathbf s}(\pmb \zeta, \mathbf{x})$ is well-defined and analytic on $\mathcal{B}_{\delta_1}\times I_{\varepsilon_1}^k$ for every  $\mathbf s$ with  $\|\mathbf s\|_1\leq \Delta-1$.
In fact, $F^{\widetilde{\varphi}}_{\mathbf s}$ is a (multivariate) analytic continuation of $F^{{\varphi}}_{\mathbf s}$.
Since $I$ is a compact interval, in the following when we pick a neighborhood $\widetilde{I}$ of $I$, without loss of generality, we may always pick $\widetilde{I}$ as an $\varepsilon$-strip $I_{\varepsilon}$ of $I$.

Then, we show that we can pick a pair of $(\delta_2, \varepsilon_2)$ where $\delta_2<\delta_1$ and $\varepsilon_2<\varepsilon_1$, a constant $M>0$ and a constant $\eta>0$ such that for every $\mathbf s$ with  $\|\mathbf s\|_1\leq \Delta-1$, we have 
$$\left\|\nabla F^{\widetilde\varphi}_{\pmb \zeta, \mathbf{s}}(\mathbf x)\right\|_1\leq1-\eta\ \text{ ~~~and ~~~ } \left\|\nabla F^{\widetilde\varphi}_{\mathbf{x}, \mathbf{s}}(\pmb{\zeta})\right\|_1\leq M $$
for all $\pmb \zeta \in \overline{\mathcal{B}_{\delta_2}}$  and all $\mathbf{x} \in \overline{I_{\varepsilon_2}^k}.$ 
By real contraction, there is an $\eta'>0$ such that $\left\|\nabla F^{\widetilde{\varphi}}_{\pmb \zeta_0, \bf s}(\mathbf{x})\right\|_1\le1-\eta'$ 
for every $\mathbf s$ with  $\|\mathbf s\|_1\leq \Delta-1$ 
and all $\mathbf{x} \in I^{k}$.
Given some $\mathbf s$ with $\|\mathbf s\|_1\leq \Delta-1$,
since $F^{\widetilde{\varphi}}_{\mathbf s}(\pmb \zeta, \mathbf{x})$ is  analytic on $\mathcal{B}_{\delta_1}\times I_{\varepsilon_1}^k$, 
by continuity we can pick some $\delta_{\mathbf s}<\delta_1$ and  $\varepsilon_{\mathbf s}<\varepsilon_1$ such that 
$\left\|\nabla F^{\widetilde{\varphi}}_{\pmb \zeta, \bf s}(\mathbf{x})\right\|_1\le1-\frac{\eta'}{2}$ 
for all $\pmb \zeta \in \overline{\mathcal{B}_{\delta_{\mathbf s}}}$  and all $\mathbf{x} \in \overline{I_{\varepsilon_{\mathbf s}}^k}.$ 
In addition, let $$M_{\mathbf s}=\sup_{\pmb \zeta \in \overline{\mathcal{B}_{\delta_{\mathbf s}}}, \mathbf{x} \in \overline{I_{\varepsilon_{\mathbf s}}^k}}\left\|\nabla F^{\widetilde\varphi}_{\mathbf{x}, \mathbf{s}}(\pmb{\zeta})\right\|_1,$$
and we know $M_{\mathbf s}<+\infty$ since $F^{\widetilde{\varphi}}$ is analytic on  $\overline{\mathcal{B}_{\delta_{\mathbf s}}}\times \overline{I_{\varepsilon_{\mathbf s}}^k}$ which is close and bounded.
Finally, let $$\eta=\frac{\eta'}{2}, ~~~
\delta_2=\min_{{\|\mathbf s\|_1}\leq\Delta-1}\{\delta_{\mathbf s}\}, ~~~
\varepsilon_2=\min_{{\|\mathbf s\|_1}\leq\Delta-1}\{\varepsilon_{\mathbf s}\},  ~~
\text{ and } ~~ M=\max_{{\|\mathbf s\|_1}\leq\Delta-1}\{M_{\mathbf s}\}.$$
These choices will satisfy our requirement.

For the case that $\|\mathbf s\|_1= \Delta$, 
we show that we can pick a pair of $(\delta_3, \varepsilon_3)$ where $\delta_3<\delta_1$ and $\varepsilon_3<\varepsilon_1$ such that for every $\mathbf s$ with  $\|\mathbf s\|_1= \Delta$, we have
$-1\notin F_{\mathbf{s}}(\overline{\mathcal{B}_{\delta_3}}, \widetilde{J}_2^k)$ 
where $\widetilde{J}_2=\widetilde{\varphi}^{-1}(\overline{I_{\varepsilon_3}})$ is a closed neighborhood of $J$.
Since $F_{\mathbf{s}}$ is analytic, and by real contraction, $-1\notin F_{\pmb \zeta_0, \bf s}(J^{k})$ which is closed.
Again by continuity we can pick some $(\delta_3, \varepsilon_3)$ that satisfy our requirement.

 Since $\pmb \zeta_0=(\beta_0, \gamma_0, \lambda_0)$ satisfies real contraction, we have $\lambda_0\in J$, $-\gamma_0\notin J$ and $-1\notin J$. Recall that $J=\widetilde{\varphi}^{-1}(I)$.
 Again, since $\widetilde{\varphi}^{-1}$ is analytic and by continuity, we can pick some $\varepsilon \leq \min\{\varepsilon_2, \varepsilon_3\}$ such that 
$\lambda_0\in \widetilde{\varphi}^{-1}(I_\varepsilon)$ (an open set), $-\gamma_0\notin \widetilde{\varphi}^{-1}(\overline{I_\varepsilon})$ (a closed set) and $-1\notin \widetilde{\varphi}^{-1}(\overline{I_\varepsilon})$. 
Moreover, we can pick some $\delta_4$ small enough such that the disc 
$B_{\lambda_0, \delta_4}:=\{z\in \mathbb{C}\mid |z-\lambda_0|<\delta_4\}$ is in $\widetilde{\varphi}^{-1}({I_\varepsilon})$,
and the disc $B_{-\gamma_0, \delta_4}:=\{z\in \mathbb{C}\mid |z-(-\gamma_0)|<\delta_4\}$ is disjoint with $\widetilde{\varphi}^{-1}(\overline{I_\varepsilon})$.
Let $P=\overline{I_\varepsilon}$ and $Q=\widetilde{\varphi}^{-1}(\overline{I_\varepsilon})$. 
Clearly, $P$ is convex.
For every  $\pmb \zeta $ with $\|\pmb \zeta-\pmb\zeta_0\|_{\infty}<\delta$,
we have $\lambda\in Q$, $-\gamma \notin Q$ and $-1\notin Q$.
In addition, we know that $Q$ is closed and bounded since $P$ is closed and bounded and $\widetilde\varphi^{-1}$ is analytic on $P$.
Finally, let $\delta=\min\{\delta_2, \delta_3, \delta_4, \frac{\varepsilon\eta}{M}\}$. 
We show that for every $\mathbf s$ with ${\|\mathbf s\|_1}\leq\Delta-1$, we have $F^{\widetilde\varphi}_{\mathbf{s}}(\mathcal{B}_{\delta}, P^k)\subseteq P$,
which implies that 
$F_{\mathbf{s}}(\mathcal{B}_{\delta}, Q^k)\subseteq Q$.

Consider some $\mathbf{x}\in P^k$. By the definition, there exists an $\mathbf{x}_0\in I^k$ such that $\|\mathbf{x}-\mathbf{x}_0\|_\infty\leq\varepsilon$.
Also, consider some $\pmb \zeta \in \mathcal{B}_{\delta}$, and we have $\|\pmb \zeta - \pmb \zeta_{0}\|_\infty<\delta$.
Then, for every $\mathbf s$ with ${\|\mathbf s\|_1}\leq\Delta-1$, consider $F^{\widetilde\varphi}_{\mathbf{s}}(\pmb \zeta, \mathbf{x})-F^{\widetilde\varphi}_{\mathbf{s}}(\pmb \zeta_0, \mathbf{x}_0).$ 
We have
\begin{equation*}
    \begin{aligned}
    &\left|F^{\widetilde\varphi}_{\mathbf{s}}(\pmb \zeta, \mathbf{x})-F^{\widetilde\varphi}_{\mathbf{s}}(\pmb \zeta_0, \mathbf{x}_0)\right|\\
    \leq  &\left|F^{\widetilde\varphi}_{\mathbf{s}}(\pmb \zeta, \mathbf{x})-F^{\widetilde\varphi}_{\mathbf{s}}(\pmb \zeta_0, \mathbf{x})\right|+
    \left|F^{\widetilde\varphi}_{\mathbf{s}}(\pmb \zeta_0, \mathbf{x})-F^{\widetilde\varphi}_{\mathbf{s}}(\pmb \zeta_0, \mathbf{x}_0)\right|\\
    \leq & \sup_{\pmb{\zeta}'\in \mathcal{B}_{\delta}}\left\|\nabla F^{\widetilde\varphi}_{\mathbf{x}, \mathbf{s}}(\pmb{\zeta}')\right\|_1\cdot \|\pmb{\zeta}-\pmb\zeta_0\|_{\infty}
    +\sup_{\mathbf{x}'\in P^k}\left\|\nabla F^{\widetilde\varphi}_{\pmb \zeta_0, \mathbf{s}}(\mathbf x')\right\|_1\cdot \left\|\mathbf{x}-\mathbf{x}_0\right\|_{\infty}\\
    \leq & M\delta+(1-\eta)\cdot \varepsilon \leq \varepsilon.\\
    \end{aligned}
\end{equation*}
The second inequality above uses the fact that both $\mathcal{B}_{\delta}$ and $P^k$ are convex, which ensures that the line between $\pmb \zeta_0$ and $\pmb \zeta$ is in $\mathcal{B}_{\delta}$ and the line between $\mathbf x_0$ and $\bf x$ is in $P^k$.
By real contraction, we know that $F^{\widetilde\varphi}_{\mathbf{s}}(\pmb \zeta_0, \mathbf{x}_0)\in I$ since $\mathbf{x}_0\in I^k$.
Thus, we have $F^{\widetilde\varphi}_{\mathbf{s}}(\pmb \zeta, \mathbf{x})\in P.$

Thus, for every  $\pmb \zeta $ with $\|\pmb \zeta-\pmb\zeta_0\|_{\infty}<\delta$,
we have $\lambda\in Q$, $-\gamma \notin Q$ and $-1\notin Q$, and
\begin{enumerate}
 \item   $F_{\pmb \zeta, \bf s}(Q^{k})\subseteq Q$ for every $\mathbf{s}$ with $\|\mathbf{s}\|_1\leq \Delta -1$ and $-1\notin F_{\pmb \zeta, \bf s}(Q^{k})$ for every $\mathbf{s}$ with $\|\mathbf{s}\|_1= \Delta$;
       \item there exists $\eta>0$ s.t. $\left\|\nabla F^\varphi_{\pmb \zeta, \bf s}(\mathbf{x})\right\|_1\le1-\eta$  for every $\mathbf{s}$ with $\|\mathbf{s}\|_1\leq \Delta -1$ and  all $\mathbf{x} \in P^{k}$.
\end{enumerate}
The function $\widetilde\varphi:Q\rightarrow P$ is a good potential function for $\pmb \zeta$.
\end{proof}



Combining Lemmas \ref{lem:c1c2}, \ref{complex-to-zerofree}, \ref{lem:complex-to-correlation} and Theorem \ref{real-to-complex}, we have the following result.
\begin{theorem}
Fix $\Delta\geq 3$.
For every $\pmb \zeta_0\in \mathcal{S}^{\Delta}_i$ $(i \in [4])$, there exists a $\delta
\footnote{The choice of $\delta$ does not depend on the size of the graph, only on $\Delta$ and $\pmb \zeta_0$.
In particular, let $D$ be a compact set in $\mathcal{S}^{\Delta}_i$ for some $i\in[4]$.
Then, there exists a uniform $\delta$ such that for all $\pmb \zeta$ in a complex neighborhood $D_{\delta}$ of radius $\delta$ around $D$, i.e.,  $\pmb\zeta \in D_{\delta}:=\{\mathbf z\in \mathbb{C}^3\mid \|\mathbf z-\mathbf{z}_0\|_{\infty}<\delta, \mathbf{z}_0\in D\}$, $Z_{G}(\pmb\zeta)\neq 0$ for every graph $G$ of degree at most $\Delta$.}
>0$ such that for any $\pmb \zeta \in \mathbb{C}^3$ where $\|\pmb\zeta-\pmb\zeta_0\|_{\infty}<\delta$, we have 
\begin{itemize}
    \item $Z^{\sigma_{\Lambda}}_{G}(\pmb\zeta)\neq 0$ for every graph $G$ of degree at most $\Delta$ and every feasible configuration $\sigma_{\Lambda}$; 
    \item the $\Delta$-bounded 2-spin system specified by $\pmb\zeta$ exhibits correlation decay.
\end{itemize}
Then via either Weitz's algorithm or Barvinok's algorithm, there is an FPTAS for computing $Z_{G}(\pmb\zeta)$.
\end{theorem}
\begin{remark}
In order to apply Barvinok's algorithm, by Lemma \ref{lem-bar}, we need to make sure that the zero-free regions contain $\lambda=0$ (an easy computing point).
This is true for $\mathcal{S}^{\Delta}_1$, $\mathcal{S}^{\Delta}_2$ and $\mathcal{S}^{\Delta}_3$. For parameters in $\mathcal{S}^{\Delta}_4$, we will reduce the problem to a case in $\mathcal{S}^{\Delta}_3$ by swapping $\beta$ and $\gamma$ and replacing $\lambda$ by $1/\lambda$. Then, one can apply Barvinok's algorithm.

\end{remark}

\section*{Acknowledgement}
The authors would like to thank Professor Jin-Yi Cai for  valuable discussions and suggestions on a preliminary version of this paper.
\renewcommand{\refname}{References}

\newpage
\appendix
\section{Appendix}
\subsection{Self-Avoiding Walk Tree}
We  adapt the description of Weitz's self-avoiding walk (SAW) tree construction from \cite{gl18} with slight modifications.
    Given a graph $G=(V,E)$ and a vertex $v\in V$,
 the SAW tree of $G$ at $v$ denoted by $T_{\text{SAW}}(G,v)$, is a tree with root $v$ that
enumerates all paths originating from $v$ in $G$.
Additional vertices closing cycles of $G$ are added as leaves of the tree  (see Figure \ref{fig:saw-tree} for an example).
Each vertex in $V$ of $G$ is mapped to some vertices in $V_{\text{SAW}}$ of $T_{\text{SAW}}(G,v)$.
For leaves in $V_{\text{SAW}}$ that close cycles,
a boundary condition is imposed.
The imposed spin of such a leaf depends on whether the orientation of the cycle is from a lower indexed vertex to a higher indexed vertex or conversely,
where the order of indices is arbitrarily chosen in $G$.
Vertex sets $S\subseteq \Lambda\subseteq V$ are mapped to  $S_{\text{SAW}}\subseteq\Lambda_{\text{SAW}}\subseteq V_{\text{SAW}}$ respectively,
and any configuration $\sigma_\Lambda\in\{0,1\}^\Lambda$ is mapped to a corresponding $\sigma_{\Lambda_{\text{SAW}}}\in\{0,1\}^{\Lambda_{\text{SAW}}}$.

	\begin{figure}[!hbtp]
\centering
	\includegraphics[scale=0.45]{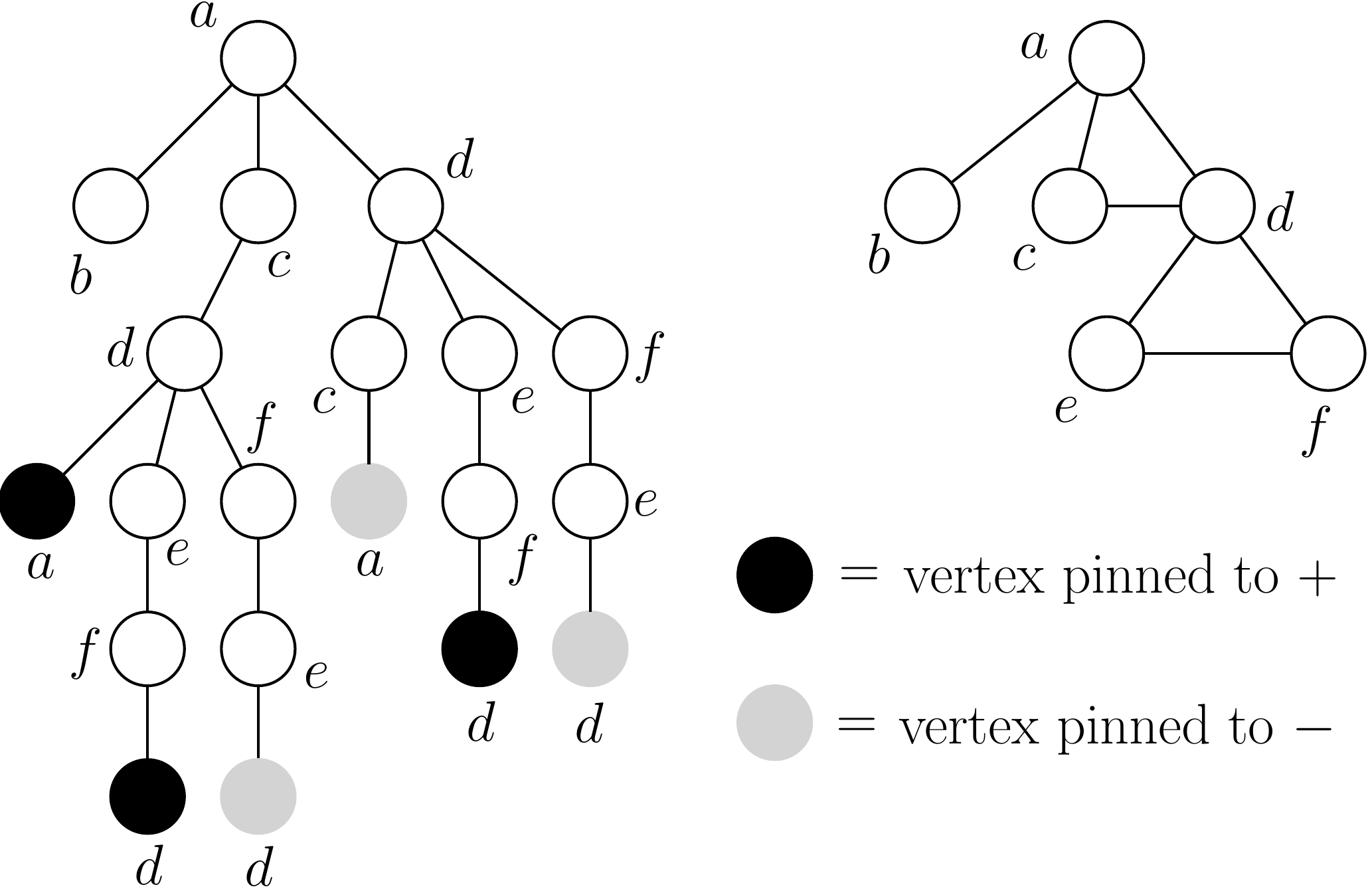}
	\caption{Weitz's SAW tree construction}
	\label{fig:saw-tree}
	\end{figure}
	
Here is the key result (Theorem 3.1 of Weitz~\cite{weitz}) for the SAW tree construction.
\begin{theorem}\label{prop:SAW}
  Let $G=(V,E)$ be a graph, $v\in V$ and $T=T_{\mathrm{SAW}}(G,v)$.
 Let $\sigma_\Lambda\in\{0,1\}^\Lambda$ be a configuration on $\Lambda\subseteq V$ where $v\notin \Lambda$,
  and $S\subseteq V$.
Then, we have
    $$R_{G,v}^{\sigma_\Lambda}(\pmb \zeta)=R_{T, v}^{\sigma_{\Lambda_{\rm{SAW}}}}(\pmb \zeta).$$
  Moreover, $\mathrm{dist}_G(v,S)=\mathrm{dist}_T(v,S_{\rm{SAW}})$, the maximum degree of $T$ is equal to the maximum degree of $G$,
   and the neighborhood of any vertex in $V_{\rm{SAW}}$ can be constructed in time proportional to the size of the neighborhood of the corresponding vertex in $V$.
\end{theorem}

\subsection{Proof of Lemma \ref{lem:c1c2}}
\begin{lemma}
Fix $\Delta\geq 3$.
For every $\pmb \zeta \in \mathcal{S}^{\Delta}_i (i\in [4])$, it satisfies real contraction for ${\Delta}$.
\end{lemma}

\begin{proof}
\vspace{-1ex}
{\bf Case 1:} $\lambda =0$.

\vspace{1ex}
We first consider a trivial case that $\lambda =0$, in which 
$F_{\pmb \zeta, \mathbf{s}}(\mathbf{x})\equiv 0$.
We pick $J=[0, 1]$ and the potential function $\varphi(x)=x$. 
Clearly, $\varphi$ is analytic  on $J$ and $\varphi'(x)=1\neq 0$ for all $x\in J$.
Also, we know $\lambda=0 \in J$, $-\gamma \notin J$ and $-1 \notin J$.
Moreover, for every $\mathbf{s}\in \mathbb{N}^3$ and all $\mathbf{x}\in J^k$,
we have $F^\varphi_{\pmb \zeta, \bf s}(\mathbf{x})=F_{\pmb \zeta, \bf s}(\mathbf{x})\equiv 0 \in J$  and $\left\|\nabla F^\varphi_{\pmb \zeta, \bf s}(\mathbf{x})\right\|_1\equiv 0$.

Thus, the function $\varphi$ defined on $J$ is a good potential function for $\pmb \zeta$.

\noindent{\bf Case 2:} $\lambda\neq 0$ and $\pmb\zeta \in \mathcal{S}_1^\Delta$.

\vspace{1ex}

Let $r=\min\{1, \beta, {1}/{\gamma}\}\leq 1$ and $t=\max\{1, \beta, {1}/{\gamma}\}\geq 1$.
We pick the interval $J=[\lambda r^{\Delta-1}, \lambda t^{\Delta-1}]$ and the potential function $\varphi=\log(x)$. Clearly, $\varphi$ is analytic on $J$ and $\varphi'(x)\neq 0$ for all $x\in J$.
Also, we know that $\lambda \in J$, $-\gamma \notin J$ and $-1 \notin J$ and $-1\notin F_{\pmb\zeta,\mathbf{s}}(J^k)$ for every $\|\mathbf{s}\|_1=\Delta$. Since $\beta>0$ and $\gamma>0$, for any $x>0$, we have 
$$r\leq \min\{\beta, 1/\gamma\}\leq  \frac{\beta x+1}{x+\gamma}\leq \max\{\beta, 1/\gamma\}\leq t.$$
Thus, for any $\mathbf{x}\in J^k$, we have
$$F_{\pmb \zeta, \mathbf{s}}(\mathbf{x})=\lambda \beta^{s_1}\gamma^{-s_2}\prod^k_{i=1}\left(\frac{\beta x_i+1}{x_i+\gamma}\right)\in \left[\lambda r^{\|\mathbf{s}\|_1}, \lambda t^{\|\mathbf{s}\|_1}\right]\subseteq \left[\lambda r^{\Delta-1}, \lambda t^{\Delta-1}\right].$$
Hence, $F_{\pmb \zeta, \mathbf{s}}(J^k)\subseteq J$ for every $\|\mathbf{s}\|_1\leq \Delta-1$.

Let $I=\varphi(J)$. Then, we consider the gradient $\nabla F^\varphi_{\pmb \zeta, \bf s}(\mathbf{x})$ for every  $\|\mathbf{s}\|_1\leq \Delta-1$ and  all $\mathbf{x}\in I^k$.
We have $$F^\varphi_{\pmb \zeta, \bf s}(\mathbf{x})=\log \lambda +s_1 \log \beta -s_2 \log \gamma+\sum_{i=1}^{k}{\log\left(\frac{\beta e^{x_i}+1}{e^{x_i}+\gamma}\right)}.$$
Thus, we have 
$$\left|\frac{\partial F^\varphi_{\pmb \zeta, \bf s}}{\partial x_i}\right|=\frac{|1-\beta\gamma|}{\beta e^{x_i}+\gamma e^{-x_i}+1+\beta\gamma}\le\frac{|1-\beta\gamma|}{2\sqrt{\beta\gamma}+1+\beta\gamma}=\frac{|1-\sqrt{\beta\gamma}|}{1+\sqrt{\beta\gamma}}.$$
Here $\beta e^{x_i}+\gamma e^{-x_i}\geq 2\sqrt{\beta\gamma}$ due to the AM-GM inequality.
Since $\frac{\Delta-2}{\Delta}<\sqrt{\beta\gamma}<\frac{\Delta}{\Delta-2}$, we have $\frac{|1-\sqrt{\beta\gamma}|}{1+\sqrt{\beta\gamma}}<\frac{1}{\Delta-1}$. We can pick an positive $\eta$ such that $\left|\frac{\partial F^\varphi_{\pmb \zeta, \bf s}}{\partial x_i}\right|\le\frac{1-\eta}{\Delta-1}$.
Therefore, we have 
$$\left\|\nabla F^\varphi_{\pmb \zeta, \bf s}(\mathbf{x})\right\|_1\leq\sum_{i=1}^k\left|\frac{\partial F^\varphi_{\pmb \zeta, \bf s}}{\partial x_i}\right|\leq\frac{(1-\eta)k}{\Delta-1}\leq\frac{(1-\eta)\|\mathbf{s}\|_1}{\Delta-1}\leq 1-\eta$$
for  every  $\|\mathbf{s}\|_1\leq \Delta-1$ and all $\mathbf{x}\in I^k$.

\vspace{1.5ex}

\noindent{\bf Case 3:} $\lambda\neq 0$ and $\pmb \zeta \in \mathcal{S}_3^\Delta$ or $\mathcal{S}_{4}^\Delta$.

\vspace{1ex}
Since $\beta \gamma>\frac{\Delta}{\Delta-2}>1$, we have $1/\gamma<\beta$.
We still pick the interval $J=[\lambda r^{\Delta-1}, \lambda t^{\Delta-1}]$ where $r=\min\{1, {1}/{\gamma}\}$ and $t=\max\{1, \beta\}$ and the the potential function $\varphi=\log(x)$.
By the same argument as in case 2, 
we know condition 1 of real contraction is satisfied. 
We need to bound the gradient  $\nabla F^\varphi_{\pmb \zeta, \bf s}(\mathbf{x})$ for  $\|\mathbf{s}\|_1\leq \Delta-1$ and  $\mathbf{x}\in I^k$ where $I=\varphi(J)$. 
Note that when $x_i\in I$, we have $e^{x_i}=\varphi^{-1}(x_i)\in J$, and $e^{-x_i}\in [\frac{1}{\lambda t^{\Delta-1}}, \frac{1}{\lambda r^{\Delta-1}}]$.

If $\pmb \zeta \in \mathcal{S}_3^\Delta$, then
we have 
$$\left|\frac{\partial F^\varphi_{\pmb \zeta, \bf s}}{\partial x_i}\right|=\frac{\beta\gamma-1}{\beta e^{x_i}+\gamma e^{-x_i}+1+\beta\gamma}\le\frac{\beta\gamma-1}{\frac{\gamma}{\lambda t^{\Delta-1}}+1+\beta\gamma}< \frac{{\beta\gamma}-1}{(\Delta-2)\beta\gamma-\Delta+1+{\beta\gamma}}=\frac{1}{\Delta-1}.$$
Otherwise, $\pmb \zeta \in \mathcal{S}_4^\Delta$. We have
$$\left|\frac{\partial F^\varphi_{\pmb \zeta, \bf s}}{\partial x_i}\right|=\frac{\beta\gamma-1}{\beta e^{x_i}+\gamma e^{-x_i}+1+\beta\gamma}\le\frac{\beta\gamma-1}{{\beta\lambda}{ r^{\Delta-1}}+1+\beta\gamma}< \frac{{\beta\gamma}-1}{(\Delta-2)\beta\gamma-\Delta+1+{\beta\gamma}}=\frac{1}{\Delta-1}.$$
Thus, in both cases, there exists some $\eta>0$ such that $\left\|\nabla F^\varphi_{\pmb \zeta, \bf s}(\mathbf{x})\right\|_1\leq1-\eta$ 
for  every  $\|\mathbf{s}\|_1\leq \Delta-1$ and all $\mathbf{x}\in I^k$. 

\vspace{1.5ex}

\noindent{\bf Case 4:} $\lambda\neq 0$ and $\pmb \zeta \in \mathcal{S}_2^\Delta$.

\vspace{1ex}

If $\beta>0$, we still pick the same interval $J$ as in case 2. We know that condition 1 of real contraction is satisfied. 

For the case that $\beta =0$, we pick  the interval $J=[\ell, m]$, where $m=\max\{\lambda,\lambda/\gamma^{\Delta-1}\}>0$ and $\ell=\min\{\lambda,\lambda/(m+\gamma)^{\Delta-1}\}>0$. 
Clearly, we have  $\lambda\in J,-\gamma\notin J,-1\notin J$, and $-1\notin F_{\pmb\zeta,\mathbf{s}}(J^k)$ for every $\|\mathbf{s}\|_1=\Delta$.
Recall that when $\beta=0$, we only consider recursion functions $F_{\pmb\zeta,\mathbf{s}}$ where $s_1=0$. 
Then for every $\mathbf{s}$ where $\|\mathbf{s}\|_1\leq \Delta-1$ and $s_1=0$ and all $\mathbf{x}\in J^k$, we have 
$$\ell\leq \frac{\lambda}{(m+\gamma)^{s_2+k}} \leq F_{\pmb\zeta,\mathbf{s}}(\mathbf{x})=\lambda\gamma^{-s_2}\prod_{i=1}^{k}{\left(\frac{1}{x_i+\gamma}\right)}\leq \frac{\lambda}{\gamma^{s_2+k}}\leq m.$$
Hence $F_{\pmb\zeta,\mathbf{s}}(J^k)\subseteq J$ for every $\|\mathbf{s}\|_1\le\Delta-1$.

Now, we pick the potential function $$\varphi(x)=\int^x_1\frac{1}{\sqrt{y(\beta y+1)(y+\gamma)}}dy.$$
Clearly, $\varphi$ is analytic on $J$ and $\varphi'(x)\neq 0$ for all $x\in J$.
We consider the gradient $\nabla F^\varphi_{\pmb \zeta, \bf s}(\mathbf{x})$.
By calculation, we have $\left\|\nabla F^\varphi_{\pmb\zeta,\mathbf{s}}(\mathbf{x})\right\|_1=H_{\pmb\zeta,\mathbf{s}}(\pmb{\varphi}^{-1}(\mathbf{x}))$ where ${\pmb \varphi^{-1}(\mathbf x)}=(\varphi^{-1}(x_1), \ldots, \varphi^{-1}(x_k))$ and
\begin{align*}
H_{\pmb\zeta,\mathbf{s}}(\mathbf{x})=(1-\beta\gamma)\cdot\sqrt{\frac{F_{\pmb\zeta,\mathbf{s}}(\mathbf{x})}{\left(\beta F_{\pmb\zeta,\mathbf{s}}(\mathbf{x})+1\right)\left(F_{\pmb\zeta,\mathbf{s}}(\mathbf{x})+\gamma\right)}}\cdot\sum_{i=1}^{k}{\sqrt{\frac{x_i}{(\beta x_i+1)(x_i+\gamma)}}}.
\end{align*}
We want to bound $\left\|\nabla F^\varphi_{\pmb\zeta,\mathbf{s}}(\mathbf{x})\right\|_1$ for every $\|\mathbf{s}\|_1\le\Delta-1$ and all $\mathbf{x}\in I^k$ where $I=\varphi(J)$, which is equivalent to bound $H_{\pmb\zeta,\mathbf{s}}(\pmb{\varphi}^{-1}(\mathbf{x}))$ for all $\pmb{\varphi}^{-1}(\mathbf{x})\in J^k\subseteq\mathbb{R}_+^k$.
We will show that there exists some $\eta>0$ such that $H_{\pmb\zeta,\mathbf{s}}(\mathbf{x})\leq 1- \eta$ for every $\|\mathbf{s}\|_1\leq \Delta-1$ and  all $\mathbf{x}\in \mathbb{R}_{+}^k$. 
This will finish the proof. Note that for any $\mathbf{s}=(s_1,s_2,k)$ with $k=0$, we have $\nabla F_{\pmb\zeta,\mathbf{s}}^\varphi(\mathbf{x})\equiv0$ and we are done. Thus, we may assume that $k\geq 1$. 
We prove our claim in two steps.

\emph{Step 1.} Let \begin{align*}
h_d(x)=d(1-\beta\gamma)\cdot\sqrt{\frac{x}{(\beta x+1)(x+\gamma)}}\cdot\sqrt{\frac{\lambda\left(\frac{\beta x+1}{x+\gamma}\right)^d}{\left(\beta\lambda\left(\frac{\beta x+1}{x+\gamma}\right)^d+1\right)\left(\lambda\left(\frac{\beta x+1}{x+\gamma}\right)^d+\gamma\right)}}
\end{align*} be the symmetrized univariate version of $H_{\pmb\zeta,\mathbf{s}}(\mathbf{x})$ where $d=\|\mathbf{s}\|_1\leq \Delta-1$.
We show that there exists some $\hat{x}>0$ such that $H_{\pmb\zeta,\mathbf{s}}(\mathbf{x})\le h_d(\hat{x})$.

Consider some $\mathbf{x}=(x_1, \ldots, x_k)\in \mathbb{R}_+^k$.
For every $x_i>0$,  let $z_i=\frac{\beta x_i+1}{x_i+\gamma}$, and we have $z_i\in\left(\beta,\frac{1}{\gamma}\right)$. 
Also, let $z_{k+1}=\cdots=z_{k+s_1}=\beta$ and $z_{k+s_1+1}=\cdots=z_d=\frac{1}{\gamma}$. 
Then, we have
\begin{align*}
H_{\pmb\zeta,\mathbf{s}}(\mathbf{x})=\sqrt{\frac{\lambda\prod_{i=1}^{d}{z_i}}{\left(\beta\cdot\lambda\prod_{i=1}^{d}{z_i}+1\right)\left(\lambda\prod_{i=1}^{d}{z_i}+\gamma\right)}}\cdot\sum_{i=1}^{d}{\sqrt{\left(\frac{1}{z_i}-\gamma\right)(z_i-\beta)}}.
\end{align*}
By Cauchy-Schwarz inequality and AM-GM inequality, we have
\begin{align*}
\sum_{i=1}^{d}{\sqrt{\left(\frac{1}{z_i}-\gamma\right)(z_i-\beta)}}\le d\sqrt{1+\beta\gamma-\frac{1}{d}\sum_{i=1}^{d}{\left(\gamma z_i+\frac{\beta}{z_i}\right)}}\le d\sqrt{1+\beta\gamma-\gamma\left(\prod_{i=1}^{d}{z_i}\right)^\frac{1}{d}-\beta\left(\prod_{i=1}^{d}{z_i}\right)^{-\frac{1}{d}}}.
\end{align*}
Let $\hat{z}=\left(\prod_{i=1}^{d}{z_i}\right)^\frac{1}{d}$. Since  $z_i \in\left(\beta,\frac{1}{\gamma}\right)$ for $i\in [k]$ and $k\geq 1$, we know $\hat{z}\in\left(\beta,\frac{1}{\gamma}\right)$. 
Then, we have
\begin{align*}
H_{\pmb\zeta,\mathbf{s}}(\mathbf{x})\le\sqrt{\frac{\lambda\hat{z}^d}{\left(\beta\lambda\hat{z}^d+1\right)\left(\lambda\hat{z}^d+\gamma\right)}}\cdot d\sqrt{1+\beta\gamma-\gamma\hat{z}-\frac{\beta}{\hat{z}}}=d\sqrt{\frac{\lambda\hat{z}^d\left(\frac{1}{\hat{z}}-\gamma\right)(\hat{z}-\beta)}{\left(\beta\lambda\hat{z}^d+1\right)\left(\lambda\hat{z}^d+\gamma\right)}}.
\end{align*}
Let $\hat{x}=\frac{1-\gamma\hat{z}}{\hat{z}-\beta}$. Then we have $\hat{x}>0$, and
\begin{align*}
H_{\pmb\zeta,\mathbf{s}}(\mathbf{x})\le k(1-\beta\gamma)\cdot\sqrt{\frac{\hat{x}}{(\beta\hat{x}+1)(\hat{x}+\gamma)}}\cdot\sqrt{\frac{\lambda\left(\frac{\beta \hat{x}+1}{\hat{x}+\gamma}\right)^d}{\left(\beta\lambda\left(\frac{\beta\hat{x}+1}{\hat{x}+\gamma}\right)^d+1\right)\left(\lambda\left(\frac{\beta\hat{x}+1}{\hat{x}+\gamma}\right)^d+\gamma\right)}}=h_d(\hat{x}).
\end{align*}

\emph{Step 2.} We show that there exists some $\eta>0$ such that $h_d(x)\le 1-\eta$  for every $1\le d\le\Delta-1$ and all $x>0$.

We characterize the point $x$ at which $h_d(x)$ achieves its maximum. Recall that we define $f_d(x)=\lambda\left(\frac{\beta x+1}{x+\gamma}\right)^d$ and we have $f'_d(x)=\frac{d(\beta\gamma-1)f_d(x)}{(\beta x+1)(x+\gamma)}$.
Consider the derivative of $h_d(x)$, we have
\begin{align*}
h'_d(x)=\frac{d(1-\beta\gamma)g'_d(x)}{2\sqrt{g_d(x)}},
\end{align*}
where $$g_d(x)=\frac{xf_d(x)}{(\beta x+1)(x+\gamma)(\beta f_d(x)+1)(f_d(x)+\gamma)},$$ and its derivative 
\begin{align*}
g'_d(x)=\frac{d(1-\beta\gamma)xf_d(x)}{(\beta x+1)^2(x+\gamma)^2(\beta f_d(x)+1)(f_d(x)+\gamma)}\cdot\left(\frac{\gamma-\beta x^2}{d(1-\beta\gamma)x}-\frac{\gamma-\beta f_d(x)^2}{(\beta f_d(x)+1)(f_d(x)+\gamma)}\right).
\end{align*}
We want to solve $h'_d(x)=0$. 
Since $1-\beta\gamma> 0$, it is equivalent to solve the equation \begin{align}\label{zero}
\frac{\gamma-\beta x^2}{d(1-\beta\gamma)x}=\frac{\gamma-\beta f_d(x)^2}{(\beta f_d(x)+1)(f_d(x)+\gamma)}.
\end{align}
Note that as $x$ increases from $0$ to $+\infty$, the function $\frac{\gamma-\beta x^2}{d(1-\beta\gamma)x}$  strictly decreases from $+\infty$ to $-\infty$. On the other hand, the function $\frac{\gamma-\beta f_d(x)^2}{(\beta f_d(x)+1)(f_d(x)+\gamma)}$ strictly increases since $f_d(x)$ strictly decreases as $x$ increases.
Therefore equation (\ref{zero})
has a unique solution in $(0,+\infty)$, denoted by $x_d$. Furthermore, we have
\begin{align}\label{derivative}
g'_d(x)
\begin{cases}
>0 & \text{if $0<x<x_d$,}\\
=0 & \text{if $x=x_d$,}\\
<0 & \text{if $x>x_d$.}
\end{cases}
\end{align}
Clearly the sign of $h'_d(x)$ is the same as that of $g'_d(x)$. Hence $h_d(x)$ achieves its maximum when $x=x_d$.
Then for any $x>0$, we have
\begin{align}\label{e1}
h_d(x)\le h_d(x_d)&=d(1-\beta\gamma)\cdot\sqrt{\frac{x_df_d(x_d)}{(\beta x_d+1)(x_d+\gamma)(\beta f_d(x_d)+1)(f_d(x_d)+\gamma)}}\notag\\&=\sqrt{\frac{d(1-\beta\gamma)f_d(x_d)(\gamma-\beta x_d^2)}{(\beta x_d+1)(x_d+\gamma)(\gamma-f_d(x_d)^2)}}.
\end{align}
Here, we substitute $(\beta f_d(x_d)+1)(f_d(x_d)+\gamma)$ by $\frac{d(1-\beta\gamma)\left(\gamma-\beta f_d(x_d)^2\right)x_d}{\gamma-\beta x_d^2}$ according to~\eqref{zero}.
Consider the function 
\begin{align*}
p_d(x):=\sqrt{\frac{d(1-\beta\gamma)f_d(x)(\gamma-\beta x^2)}{(\beta x+1)(x+\gamma)(\gamma-f_d(x)^2)}}.
\end{align*}
Then, we have $h_d(x)\leq p_d(x_d)$ for any $x>0$.
Now, we claim that for any $1\le d\le\Delta-1$,
\begin{align}\label{e2}
p_d(x_d)\le p_d(\hat{x}_d),
\end{align}
where $\hat{x}_d$ is the unique positive fixed point of $f_d(x)$. To prove the above claim, we only need to show that $p_d(x)$ is decreasing if $\hat{x}_d\le x_d$ and increasing if $\hat{x}_d>x_d$.
\begin{itemize}
\item If $\hat{x}_d\le x_d$, we will  show that $p_d(x)$ is decreasing on the range $[\hat{x}_d,x_d]$. By~\eqref{derivative}, we know  $g'_d(\hat{x}_d)\ge0$. Note that
\begin{align*}
g'_d(\hat{x}_d)=\frac{d(1-\beta\gamma)(\gamma-\beta\hat{x}_d^2)\hat{x}_d^2}{(\beta\hat{x}_d+1)^3(\hat{x}_d+\gamma)^3}\cdot\left(\frac{1}{d(1-\beta\gamma)\hat{x}_d}-\frac{1}{(\beta\hat{x}_d+1)(\hat{x}_d+\gamma)}\right).
\end{align*}
Since $\pmb \zeta$ is up-to-$\Delta$ unique, we have $\left|f'_d(\hat{x}_d)\right|=\frac{d(1-\beta\gamma)\hat{x}_d}{(\beta\hat{x}_d+1)(\hat{x}_d+\gamma)}<1$ and hence $\frac{1}{d(1-\beta\gamma)\hat{x}_d}-\frac{1}{(\beta\hat{x}_d+1)(\hat{x}_d+\gamma)}>0$. 
Thus, we have $\gamma-\beta\hat{x}_d^2\ge0$.  Also, since $f_d(x)$  strictly decreases as $x$ increases, we have
\begin{align*}
\gamma-\beta f_d(x_d)^2\ge\gamma-\beta f_d(\hat{x}_d)^2=\gamma-\beta\hat{x}_d^2\ge0.
\end{align*}
Then by equality~\eqref{zero}, $\gamma-\beta x_d^2$ and $\gamma-\beta f_d(x_d)^2$ must be both  positive or negative. Thus we have $\gamma-\beta x_d^2\ge0$. Then both $\frac{\gamma-\beta x^2}{(\beta x+1)(x+\gamma)}$ and $\frac{f_d(x)}{\gamma-\beta f_d(x)^2}$ are positive and strictly decreasing on $[\hat{x}_d,x_d]$.
Thus, $p_d(x)$ is strictly decreasing on $[\hat{x}_d,x_d]$, and hence $p_d(x_d)\le p_d(\hat{x}_d)$.
\item Otherwise, $\hat{x}_d>x_d$. By a similar argument as the above,
we have $\gamma-\beta f_d(\hat{x}_d)^2=\gamma-\beta\hat{x_d}^2<0,\gamma-\beta f_d(x_d)^2<0$ and $\gamma-\beta x_d^2<0$.
Hence both $\frac{\gamma-\beta x^2}{(\beta x+1)(x+\gamma)}$ and $\frac{f_d(x)}{\gamma-\beta f_d(x)^2}$ are negative and strictly decreasing on $[x_d,\hat{x}_d]$ and hence their product is positive and increasing on $[x_d,\tilde{x}_d]$. Thus we have $p_d(x)$ is increasing, and $p_d(x_d)\le p_d(\hat{x}_d)$.
\end{itemize}
Combining	~\eqref{e1} and~\eqref{e2}, we have  for all $x>0$,
\begin{align*}
h_d(x)\le h_d(x_d)=p_d(x_d)\le p_d(\hat{x}_d)=\sqrt{\frac{d(1-\beta\gamma)\hat{x}_d}{(\beta\hat{x}_d+1)(\hat{x}_d+\gamma)}}=\sqrt{\left|f'_d(\hat{x}_d)\right|}.
\end{align*}
Since $\pmb \zeta$ is up-to-$\Delta$ unique, there exists a constant $0<c<1$ such that $\left|f'_d(\hat{x}_d)\right|\le c$ for every integer $1\le d\le\Delta-1$.
Let $\eta=1-\sqrt{c}>0$.
Then, we have $h_d(x)\leq 1-\eta$ for all $x>0$.
\end{proof}

\subsection{Proof of Lemma \ref{complex-to-zerofree}}
\begin{lemma}
If $\pmb \zeta$ satisfies complex contraction for $\Delta$, then $Z_{G}^{\sigma_{\Lambda}}(\pmb \zeta)\neq 0$ for any graph $G$ of degree at most $\Delta$ and any feasible configuration $\sigma_{\Lambda}$.
\end{lemma}
\begin{proof}
If $\lambda=0$, then we have $Z_G^{\sigma_\Lambda}(\pmb\zeta)=\gamma^{|E|}\ne0$ for any graph $G=(V,E)$ and any feasible configuration $\sigma_\Lambda$. Hence we assume that $\lambda\ne0$ in the rest of the proof.

Let $t=t(G, \sigma_\Lambda)$ be the number of free vertices of a graph $G=(V,E)$ with a configuration $\sigma_\Lambda$, i.e., $t=|V|-|\Lambda|$.
We prove this lemma by induction on $t$.
For the base case $t(G, \sigma_\Lambda)=0$, we know all vertices of $G$ are pinned by $\sigma_\Lambda$.  Since $\sigma_\Lambda$ is feasible, we have $Z_G^{\sigma_\Lambda}(\pmb\zeta)\ne0$.

Now suppose that for some nonnegative integer $n$, 
it holds that $Z^{\tau_{\Lambda'}}_{G'}(\pmb\zeta)\ne0$ 
for any graph $G'$ of degree at most $\Delta$ and any feasible configuration $\tau_{\Lambda'}$ 
where $t(G', \tau_{\Lambda'})\leq n$.
We consider an arbitrary graph $G$ of degree at most $\Delta$ and 
a feasible configuration $\sigma_\Lambda$ where $t(G, \sigma_\Lambda)=n+1$.
We show that $Z_G^{\sigma_\Lambda}(\pmb\zeta)\ne0$.
We pick a free vertex $v$ in $G$. By the induction hypothesis, we have $Z_{G,v}^{\sigma_\Lambda,-}(\pmb\zeta)\neq 0$ since we further pinned one vertex of $G$ to spin $-$. 
Thus,  the ratio $R^{\sigma_\Lambda}_{G,v}=\frac{Z^{\sigma_\Lambda,+}_{G,v}(\pmb\zeta)}{Z^{\sigma_\Lambda,-}_{G,v}(\pmb\zeta)}$ is well-defined and 
it can be computed by recursion via SAW tree.
Let $T$ be the corresponding SAW tree where $v$ is the root. 
There exists an $\mathbf{s}$ with $\|\mathbf{s}\|_1\leq \Delta$ such that $$R^{\sigma_\Lambda}_{G,v}=R^{\sigma_\Lambda}_{T,v}=F_{\pmb \zeta, \mathbf{s}}\left( R^{\sigma_\Lambda}_{T_1,v_1}, \ldots, R^{\sigma_\Lambda}_{T_k,v_k}\right)$$
where $v_1, \ldots, v_k$ are free vertices of the children of $v$ and $T_1, \ldots, T_k$ are the corresponding subtrees rooted at them.
Note that in $T$, only $v$ may have $\Delta$ many children, while other nodes have at most $\Delta-1$ many children.
Therefore, for any node $v'\neq v$ and the subtree rooted at $v'$, 
the ratio $R^{\sigma_\Lambda}_{T',v'}$ can be computed by some recursion function $F_{\pmb \zeta, \mathbf{s}_{v'}}$ with $\|\mathbf{s}_{v'}\|_1\leq \Delta-1$. 
Clearly, for any free vertex $v_\ell$ at the leaf of $T$, we have $R^{\sigma_\Lambda}_{T_\ell,v_\ell}=\lambda \in Q$, where $T_\ell$ is a tree of only one vertex $v_\ell$.
By complex contraction, we have $F_{\pmb \zeta, \mathbf{s}}(Q^k)\subseteq Q$ for every $\mathbf{s}$ with $\|\mathbf{s}\|_1\leq \Delta-1$. 
By iteration on each subtree $T_i$, 
we have  $R^{\sigma_\Lambda}_{T_i,v_i}\in Q$ for every $i\in [k]$.
Also by complex contraction, we have $-1\notin F_{\pmb \zeta, \mathbf{s}}(Q^k)$ for every $\mathbf{s}$ with $\|\mathbf{s}\|_1\leq \Delta$.
Thus, we have $R^{\sigma_\Lambda}_{G,v}\neq-1$. This implies that $Z_{G}^{\sigma_\Lambda}(\pmb\zeta)\neq 0$.
\end{proof}

\subsection{Proof of Lemma \ref{lem:complex-to-correlation}}
\begin{lemma}
If $\pmb{\zeta}$ satisfies complex contraction for $\Delta$, then the $\Delta$-bounded 2-spin system specified by $\pmb{\zeta}$
exhibits SSM (correlation decay).
\end{lemma}
\begin{proof}
By Lemma~\ref{complex-to-zerofree}, we know condition 1 of SSM (Definition~\ref{def:correlation-decay}) is satisfied. We only need to show that condition 2 is satisfied. 
If $\lambda=0$, then we have $p_v^{\sigma_{\Lambda}}\equiv 0$ for any feasible configuration $\sigma_{\Lambda}$ and SSM holds trivially. 
Thus, we assume that $\lambda\ne0$.
By Weitz's SAW tree construction, we only need to show that the 2-spin systems on trees of degree at most $\Delta$ exhibits SSM.

Let $\varphi:P\to Q$ be a good potential function for $\pmb\zeta$. Let $T=(V, E)$ be a tree of degree at most $\Delta$ and $v$ be the root of $T$. 
Consider two feasible configurations $\sigma_{\Lambda_1}$ and $\tau_{\Lambda_2}$ on $\Lambda_1\subseteq V$ and $\Lambda_2\subseteq V$ respectively where $v \notin \Lambda_1 \cup \Lambda_2$. 
We want to show that  $\big|p_v^{\sigma_{\Lambda_1}}-p_v^{\tau_{\Lambda_2}}\big|\le\exp\left(-\Omega\left(\mathrm{dist}_T(v,S)\right)\right)$, where $S\subseteq\Lambda_1 \cup \Lambda_2$ is the subset on which $\sigma_{\Lambda_1}$ and $\tau_{\Lambda_2}$ differ.
Note that all vertices in $T$ except the root $v$ have at most $\Delta-1$ many children. 
We first consider the case that $v$ has at most $\Delta-1$ many children. 
Let $t=\mathrm{dist}_T(v,S)$.
We will show  $$\left|\varphi\big(R^{\sigma_{\Lambda_1}}_{T,v}\big)-\varphi\big(R^{\tau_{\Lambda_2}}_{T,v}\big)\right|\le C(1-\eta)^{t-1}$$ for some constant $C, \eta>0$ by induction on $t$.

 For the base case $t=1$, 
 since $\pmb\zeta$ satisfies complex contraction, we have $R^{\sigma_{\Lambda_1}}_{T,v}, R^{\tau_{\Lambda_2}}_{T,v}\in Q$ and hence $\varphi\big(R_{T,v}^{\sigma_{\Lambda_1}}\big),\varphi\big(R_{T,v}^{\tau_{\Lambda_2}}\big)\in P$.
 Let $C=\sup_{{z}_1, {z}_2 \in P}|z_1-z_2|.$
 Since $P$ is a closed and bounded region, we know $C<+\infty$. 
 Clearly, we have $\left|\varphi\big(R^{\sigma_{\Lambda_1}}_{T,v}\big)-\varphi\big(R^{\tau_{\Lambda_2}}_{T,v}\big)\right|\le C$.
 
Since $\pmb \zeta$ satisfies complex contraction, let $\eta$ be the constant such that $\left\|\nabla F^\varphi_{\pmb \zeta, \bf s}(\mathbf{x})\right\|_1\le1-\eta$  for every $\mathbf{s}$ with $\|\mathbf{s}\|_1\leq \Delta -1$ and  all $\mathbf{x} \in P^{k}$.
 Suppose that $\left|\varphi\big(R^{\sigma_{\Lambda_1}}_{T,v}\big)-\varphi\big(R^{\tau_{\Lambda_2}}_{T,v}\big)\right|\le C(1-\eta)^{t-1}$ for $t\leq n$ where $n$ is a positive integer. 
 We consider $t=n+1$. Since $t>1$, the configurations of all children of $v$ are the same in both $\sigma_{\Lambda_1}$ and $\tau_{\Lambda_2}$. Suppose that $v$ has $d$ children, and in both configurations $\sigma_{\Lambda_1}$ and $\tau_{\Lambda_2}$, $s_1$ of them are pinned to $+$, $s_2$ are pinned to  $-$, and $k$ are free. We denote these $k$ free vertices by $v_i$ $(i\in[k])$.
Let $T_i$ be the corresponding subtree rooted at $v_i$, and 
$\sigma^i_{\Lambda_1}$ and $\tau^i_{\Lambda_2}$ denote the configurations $\sigma_{\Lambda_1}$ and $\tau_{\Lambda_2}$ restricted on subtree $T_i$ respectively. 
Let $x_i=\varphi\Big(R_{T_i,v_i}^{\sigma_{\Lambda_1}^i}\Big)$ and $y_i=\varphi\Big(R_{T_i,v_i}^{\tau_{\Lambda_2}^i}\Big)$. 
  Since $\pmb \zeta$ satisfies complex contraction, same as we showed in the proof of Lemma \ref{complex-to-zerofree}, we have $R_{T_i,v_i}^{\sigma_{\Lambda_1}^i}, R_{T_i,v_i}^{\tau_{\Lambda_2}^i}\in Q$ and hence $x_i, y_i\in P$.
  Let $S_i=S\cap T_i$.
  Clearly, we have ${\rm dist}_{T_i}(v_i, S_i)\ge{\rm dist}_{T}(v, S)-1=t-1=n$.
 By induction hypothesis, we have $|x_i-y_i|\leq C(1-\eta)^{{\rm dist}_{T_i}(v_i, S_i)-1}\leq C(1-\eta)^{n-1}.$ 
 Then, we have
\begin{align*}
\left|\varphi\big(R_{T,v}^{\sigma_{\Lambda_1}}\big)-\varphi\big(R_{T,v}^{\tau_{\Lambda_2}}\big)\right|&=\left|\varphi\left(F_{\pmb\zeta,\mathbf{s}}\left(R_{T_1,v_1}^{\sigma^1_{\Lambda_1}},\ldots,R_{T_k,v_k}^{\sigma^k_{\Lambda_1}}\right)\right)-\varphi\left(F_{\pmb\zeta,\mathbf{s}}\left(R_{T_1,v_1}^{\tau^1_{\Lambda_2}},\ldots,R_{T_k,v_k}^{\tau^k_{\Lambda_2}}\right)\right)\right|\\
&=\left|F^\varphi_{\pmb\zeta,\mathbf{s}}(x_1,\ldots,x_k)-F^\varphi_{\pmb\zeta,\mathbf{s}}(y_1,\ldots,y_k)\right|\\&\le\sup_{\mathbf{z}\in P^k}{\left\|\nabla F_{\pmb\zeta,\mathbf{s}}^\varphi(\mathbf{z})\right\|_1}\cdot\left\|\mathbf{x}-\mathbf{y}\right\|_\infty\\&\le(1-\eta)\cdot C(1-\eta)^{n-1}\\&=C(1-\eta)^{t-1},
\end{align*}
where the first inequality is due to the fact that $P$ is convex.

We are going to bound $\left|R^{\sigma_{\Lambda_1}}_{T,v}-R^{\tau_{\Lambda_2}}_{T,v}\right|$.
Define functions $G_{\pmb \zeta, \mathbf{s}}(\mathbf{x}):=F_{\pmb \zeta, \mathbf{s}}(\pmb \varphi^{-1}(\mathbf{x})).$
Since $F_{\pmb \zeta, \mathbf{s}}$ is analytic on $Q^k$, we have $G_{\pmb \zeta, \mathbf{s}}$ is well-defined and analytic on $P^k$ for all $\|\mathbf{s}\|_1\le\Delta$.
For every $\mathbf{s}$ with $\|\mathbf{s}\|_1\leq \Delta$, 
let $Q_\mathbf{s}=G_{\pmb \zeta, \mathbf{s}}(P^k)=F_{\pmb \zeta, \mathbf{s}}(Q^k)$ and $M_ \mathbf{s}=\sup_{\mathbf{x}\in P^k}\left\|\nabla G_{\pmb \zeta, \mathbf{s}}(\mathbf{x})\right\|_1.$
Since $P^k$ is compact (closed and bounded), we have $Q_\mathbf{s}$ is compact and $M_\mathbf{s}<+\infty$.
Finally, let $$Q'=\bigcup_{\|\mathbf{s}\|_1\leq \Delta}Q_\mathbf{s}, ~~~ C'=\sup_{z_1, z_2\in Q'}|z_1-z_2|, ~~~ \text{ and } ~~~ M=\max_{\|\mathbf{s}\|_1\leq \Delta}M_\mathbf{s}.$$
Since there is only a finite number of $\mathbf{s}$ such that $\|\mathbf{s}\|_1\leq \Delta$, we have $M<+\infty$ and $Q'$ is compact, and hence $C'<+\infty$. Let $N=\max\{C',CM\}$.
We show that $\left|R^{\sigma_{\Lambda_1}}_{T,v}-R^{\tau_{\Lambda_2}}_{T,v}\right|\leq N(1-\eta)^{t-2}$.

If $t=1$, then there exist $\mathbf{s}_1$ and $\mathbf{s}_2$ where $\|\mathbf{s}_1\|_1, \|\mathbf{s}_2\|_1 \leq \Delta$ such that $R_{T,v}^{\sigma_{\Lambda_1}}\in F_{\pmb \zeta, \mathbf{s}_1}(Q^{k_1})$ and $R_{T,v}^{\tau_{\Lambda_2}}\in F_{\pmb \zeta, \mathbf{s}_2}(Q^{k_2})$.
Then, we have $R_{T,v}^{\sigma_{\Lambda_1}}, R_{T,v}^{\tau_{\Lambda_2}}\in Q'$, and hence 
\begin{align*}
\left|R_{T,v}^{\sigma_{\Lambda_1}}-R_{T,v}^{\tau_{\Lambda_2}}\right|\le C'\le \frac{N}{1-\eta}.
\end{align*}
Otherwise $t>1$. The configurations of all children of $v$ are the same in both $\sigma_{\Lambda_1}$ and $\tau_{\Lambda_2}$. 
Again, let $x_i=\varphi\Big(R_{T_i,v_i}^{\sigma_{\Lambda_1}^i}\Big)$ and $y_i=\varphi\Big(R_{T_i,v_i}^{\tau_{\Lambda_2}^i}\Big)$, where $v_i$, $T_i$, $\sigma_{\Lambda_1}^i$, $\tau_{\Lambda_2}^i$ and $S_i$ $(i\in[k])$ are all defined the same as in the above induction proof. 
Since in the subtree $T_i$, the root $v_i$ has at most $\Delta-1$ many children, 
we have $$|x_i-y_i|\leq C(1-\eta)^{{\rm dist}_{T_i}(v_i, S_i)-1}\leq C(1-\eta)^{t-2}.$$ 
Then, similarly as we did in the above induction proof, we have 
\begin{align*}
\left|R_{T,v}^{\sigma_{\Lambda_1}}-R_{T,v}^{\tau_{\Lambda_2}}\right|&=\left|F_{\pmb\zeta,\mathbf{s}}\left(R_{T_1,v_1}^{\sigma_{\Lambda_1}^1},\ldots,R_{T_k,v_k}^{\sigma_{\Lambda_1}^k}\right)-F_{\pmb\zeta,\mathbf{s}}\left(R_{T_1,v_1}^{\tau_{\Lambda_2}^1},\ldots,R_{T_k,v_k}^{\tau_{\Lambda_2}^k}\right)\right|\\&=\left|G_{\pmb\zeta,\mathbf{s}}(x_1,\ldots,x_k)-G_{\pmb\zeta,\mathbf{s}}(y_1,\ldots,y_k)\right|\\&\le\sup_{\mathbf{x}\in P^k}\left\|\nabla G_{\pmb \zeta, \mathbf{s}}(\mathbf{x})\right\|_1\cdot\|\mathbf{x}-\mathbf{y}\|_\infty\\&\le M\cdot C(1-\eta)^{t-2}\\&\le N(1-\eta)^{t-2}.
\end{align*}

Finally, we bound $\left|p_v^{\sigma_{\Lambda_1}}-p_v^{\tau_{\Lambda_2}}\right|$ from $\left|R_{T,v}^{\sigma_{\Lambda_1}}-R_{T,v}^{\tau_{\Lambda_2}}\right|$. Let $K=\inf_{z\in Q'}|1+z|$. By the complex contraction property,  $-1\notin Q_{\mathbf{s}}$ for every $\|\mathbf{s}\|_1\leq\Delta$ and thus $-1\notin Q'$.
Also, since $Q'$ is compact, we have $K>0$. Hence,
\begin{align*}
\left|p_v^{\sigma_{\Lambda_1}}-p_v^{\tau_{\Lambda_2}}\right|&=\left|\frac{R_{T,v}^{\sigma_{\Lambda_1}}}{1+R_{T,v}^{\sigma_{\Lambda_1}}}-\frac{R_{T,v}^{\tau_{\Lambda_2}}}{1+R_{T,v}^{\tau_{\Lambda_2}}}\right|=\frac{\left|R_{T,v}^{\sigma_{\Lambda_1}}-R_{T,v}^{\tau_{\Lambda_2}}\right|}{\left|1+R_{T,v}^{\sigma_{\Lambda_1}}\right|\cdot\left|1+R_{T,v}^{\tau_{\Lambda_2}}\right|}\\&\le\frac{1}{K^2}\cdot\left|R_{T,v}^{\sigma_{\Lambda_1}}-R_{T,v}^{\tau_{\Lambda_2}}\right|\le\frac{N(1-\eta)^{t-2}}{K^2}\\&\le\exp\left(-\Omega(\mathrm{dist}_T(v,S)\right).
\end{align*}
\end{proof}
\bibliography{main}{}
\bibliographystyle{alpha}
\end{document}